\documentclass[lettersize,journal]{IEEEtran}
\usepackage{textcomp}
\usepackage{url}
\usepackage{verbatim}
\usepackage{etoolbox}
 \usepackage{amsmath,amssymb}
 \usepackage{graphicx,graphics,color,psfrag}
 \usepackage{cite,balance}
 \allowdisplaybreaks[4]
 \usepackage{algorithm}
 \usepackage{accents}
 \usepackage{amsthm}
 \usepackage{bm}
 \usepackage{algorithmic}
 \usepackage[english]{babel}      
 \usepackage{multirow}
 \usepackage{enumerate}
 \usepackage{cases}
 \usepackage{stfloats}
 \usepackage{dsfont}
 \usepackage{color,soul}
 \usepackage{amsfonts}
 \usepackage{cite,graphicx,amsmath,amssymb}
 \usepackage{fancyhdr}
 \usepackage{hhline}
 \usepackage{graphicx,graphics}
 \usepackage{array,color}
 \usepackage{amsmath}
\usepackage{float}
\usepackage{amssymb}
\usepackage{amsmath}
\usepackage{amsthm}
\usepackage{amsfonts}
\usepackage{graphicx}
\usepackage{algorithm}
\usepackage{algorithmic}
\usepackage{epstopdf}
\usepackage{cite}
\usepackage{amsmath,bm}
\usepackage{subfigure}
\usepackage{graphicx}
\usepackage{color}
\usepackage{graphicx}
\usepackage{calc}
\usepackage{diagbox}

\newtheorem{theorem}{\textbf{Theorem}}

\newtheorem{proposition}{\textbf{Proposition}}

\newtheorem{assumption}{Assumption}

\begin{document}
\bibliographystyle{IEEEtran}

\title{CDDM: Channel Denoising Diffusion Models for Wireless Semantic Communications}

\author{Tong Wu, Zhiyong Chen, \emph{Senior Member, IEEE}, Dazhi He, \emph{Member, IEEE}, Liang Qian, Yin Xu, \emph{Member, IEEE}, Meixia Tao, \emph{Fellow, IEEE}, and Wenjun Zhang, \emph{Fellow, IEEE}
\thanks{ The paper will be presented in part at IEEE GLOBECOM 2023 \cite{gc-cddm}.}

\thanks{The authors are with the Cooperative Medianet Innovation Center (CMIC), Shanghai Jiao Tong University, Shanghai 200240, China, and Shanghai Key Laboratory of Digital Media Processing and Transmission (e-mail: \{wu\_tong, zhiyongchen, hedazhi, lqian, xuyin, mxtao, zhangwenjun\}@sjtu.edu.cn). L. Qian, M. Tao and W. Zhang are also with the Department of Electronic Engineering, Shanghai Jiao Tong University, Shanghai 200240, China.}}
		
\maketitle

\begin{abstract}
Diffusion models (DM) can gradually learn to remove noise, which have been widely used in artificial intelligence generated content (AIGC) in recent years. The property of DM for eliminating noise leads us to wonder whether DM can be applied to wireless communications to help the receiver mitigate the channel noise. To address this, we propose channel denoising diffusion models (CDDM) for semantic communications over wireless channels in this paper. CDDM can be applied as a new physical layer module after the channel equalization to learn the distribution of the channel input signal, and then utilizes this learned knowledge to remove the channel noise. We derive corresponding training and sampling algorithms of CDDM according to the forward diffusion process specially designed to adapt the channel models and theoretically  prove that the well-trained CDDM can effectively reduce the conditional entropy of the received signal under small sampling steps. Moreover, we apply CDDM to a semantic communications system based on joint source-channel coding (JSCC) for image transmission. Extensive experimental results demonstrate that CDDM can further reduce the mean square error (MSE) after minimum mean square error (MMSE) equalizer, and the joint CDDM and JSCC system achieves better performance than the JSCC system and the traditional JPEG2000 with low-density parity-check (LDPC) code approach.
\end{abstract}


\section{Introduction}
Diffusion models (DM)\cite{DDPM2015, Ho,Song} have recently achieved unprecedented success in artificial intelligence generated content (AIGC) \cite{Yang2022DiffusionMA}, including multimodal image generation and edition \cite{meng,Choi}, text, and video generation \cite{Lin,Sihyun}. DM is a class of latent variable models inspired by non-equilibrium thermodynamics. They directly model the score function of the likelihood function through variational lower bounds, resulting in advanced generative performance. Compared to previous generative models such as variational auto-encoder (VAE) \cite{VAE}, generative adversarial network (GAN) \cite{GAN}, and normalization flow (NF) \cite{NF}, DM can learn fine-grained knowledge of the distribution, allowing it to generate contents with rich details. Additionally, diffusion models are capable of generating more diverse images and have been shown to be resistant to mode collapse. The emergence of implicit classifiers endows diffusion models with flexibility controllability, enhanced efficiency and ensuring faithful generation in conditional generation tasks.

More specifically, DM gradually adds Gaussian noise to the available training data in the forward diffusion process until the data becomes pure noise. Then, in the reverse sampling process, it learns to recover the data from the noise, as shown in Fig. \ref{Diff}. Generally, given a data distribution  $\mathbf{x}_0\sim q(\mathbf{x}_0)$, the forward diffusion process generates the $t$-th sample of $\mathbf{x}_t$ by sampling a Gaussian vector $\epsilon \sim \mathcal{N}(0,\mathbf{I})$ as following
\begin{align}\label{DDPMforward}
  \mathbf{x}_t=\sqrt{\bar{\alpha}_t}\mathbf{x}_0+\sqrt{1-\bar{\alpha}_t}\epsilon,
\end{align}
where $\bar{\alpha}_t= {\textstyle \prod_{i=1}^{t}}\alpha_i $ and $\alpha_i\in(0,1)$ are hyperparameters.

In wireless communications, it is well known that the received signal $y$ is a noisy and distorted version of the transmitted signal $x$, e.g., we have the following for the additive white Gaussian noise (AWGN) channel
\begin{align}
  y=x+n,
\end{align}
where $n$ is white Gaussian noise.

Interestingly, compared to (1) and (2), we can find that the designed process of DM and the wireless communications system are similar. DM progressively learns to effectively remove noise, thereby generating data that closely resembles the original distribution, while the receiver in the wireless communications system aims to recover the transmitted signal from the received signal. Clearly, \textbf{can DM be applied to the wireless communications system to help the receiver remove noise?} 
\begin{figure}[t]
  \begin{center}
    \includegraphics*[width=8.8cm]{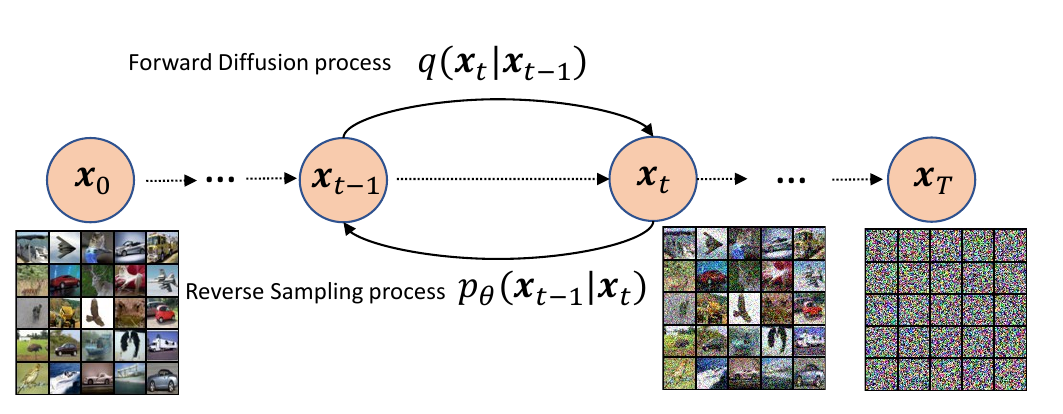}
  \end{center}
    \caption{{The forward diffusion process with transition kernel $q(\mathbf{x}_t|\mathbf{x}_{t-1})$ and the reverse sampling process with learnable transition kernel $p_\theta(\mathbf{x}_{t-1}|\mathbf{x}_t)$  of diffusion model in \cite{Ho}.}}
    \label{Diff}
\end{figure}

Motivated by this, in this paper, we design channel denoising diffusion models (CDDM) for semantic communications in the wireless communications system. The proposed CDDM is conditioned on the received signal and channel estimation results to eliminate channel noise. In contrast to conventional generative models that only generate data adhering to the original data distribution, CDDM directly generates data that closely resembles the transmitted signal $\mathbf{x}$, consequently enhancing the performance of the communications system. By employing carefully designed forward diffusion and reverse sampling processes based on an explicit conditional probabilistic model of the received signal, CDDM can adapt to different channel conditions, such as AWGN channel and Rayleigh fading channel with different signal-to-noise ratios (SNR). To leverag the received signal, we starting the reverse sampling process from the received signal rather than pure noise, greatly reducing the number of reverse sampling steps and thus accelerating the process.

In contrast to the extensive research on DM in AIGC, there have been few works of DM in wireless communications so far. In \cite{kim}, DM is employed to generate the wireless channel for an end-to-end communications system, achieving almost the same performance as the channel-aware case. In \cite{Yoni}, DM with an adapted diffusion process is proposed for the decoding of algebraic block codes. Additionally, \cite{Eleonora} applies DM as the semantic decoder to generate the image condition on the transmitted semantic segment labels of the original image, achieving excellent mean intersection over union (mIoU) and learned perceptual image patch similarity (LPIPS) performance.

On the other hand, semantic communications \cite{Lan,jihong} have emerged as a novel paradigm that facilitates the seamless integration of information and communication technology with artificial intelligence (AI), which have been recognized as a highly promising solution for the sixth-generation (6G) wireless networks \cite{shen2}. Semantic communications emphasize the transmission of valuable semantic information rather than bits, thereby guaranteeing improved transmission efficiency and reliability. One fundamental concept behind semantic communications is to bridge the source and channel components of Shannon theory \cite{Shannon}, thereby enhancing the overall performance of end-to-end transmission. The paradigm focusing on the integrated design of source and channel coding processing is known as joint source-channel coding (JSCC), which is a classical subject in the coding theory and information theory \cite{Fresia,Guyader,Chen}. However, traditional JSCC techniques are predominantly rooted in complex and explicit probabilistic models, heavily relying on expert manual designs which often face challenges when dealing with complex sources. Moreover, these JSCC techniques overlook semantic aspects and lack optimization for specific tasks or human visual perception. 

Many previous studies investigate deep-learning based JSCC techniques for semantic communications \cite{gundu2019, chenwei,XuTung,Dai2,KeYang,Bo}. Most studies concentrate on designing specific frameworks for different data modals and have achieved better performance compared with traditional wireless transmission schemes. For wireless image transmission, \cite{chenwei} proposes a novel JSCC method based on attention mechanisms, which can automatically adapt to various channel conditions. In \cite{Dai2}, an entropy model is proposed to achieve adaptive rate control for deep learning based JSCC architecture for semantic communications. In \cite{KeYang}, the swin transformer \cite{Ze} is integrated into the deep JSCC framework to improve the performance of wireless image transmission. \cite{Bo} develops a joint coding-modulation method and achieves end-to-end digital semantic communication system for image transmission, outperforming the analog-based JSCC system at low SNRs. Generally, the deep-learning based JSCC have shown great performance surpassing classic separation-based JPEG2000 source coding and advanced low-density parity-check (LDPC) channel coding, especially for small size images and under human visual perception evaluation matric such as muti-scale structure similarity index measure (MSSSIM) \cite{MSSSIM}. 

Despite its great potential, previous studies predominantly concentrate on the development of a more sophisticated model architecture with increased capacity to enhance overall performance. The channel distortion is handled through direct end-to-end optimization. In this case, the JSCC models solely learn coding and decoding strategies by utilizing received signal samples, combating channel interference. To more effectively mitigate channel interference, we integrate the CDDM with the JSCC-based semantic communications system for wireless image transmission, where the signal after CDDM is fed into the JSCC decoder to recover the image. As previously discussed, our CDDM is specially developed to mitigate channel distortion by eliminating channel noise based on an explicit probability of the received signal, thereby improving the performance of the JSCC-based semantic communication system.

\begin{figure*}[t]
  \begin{center}
    \includegraphics*[width=18cm]{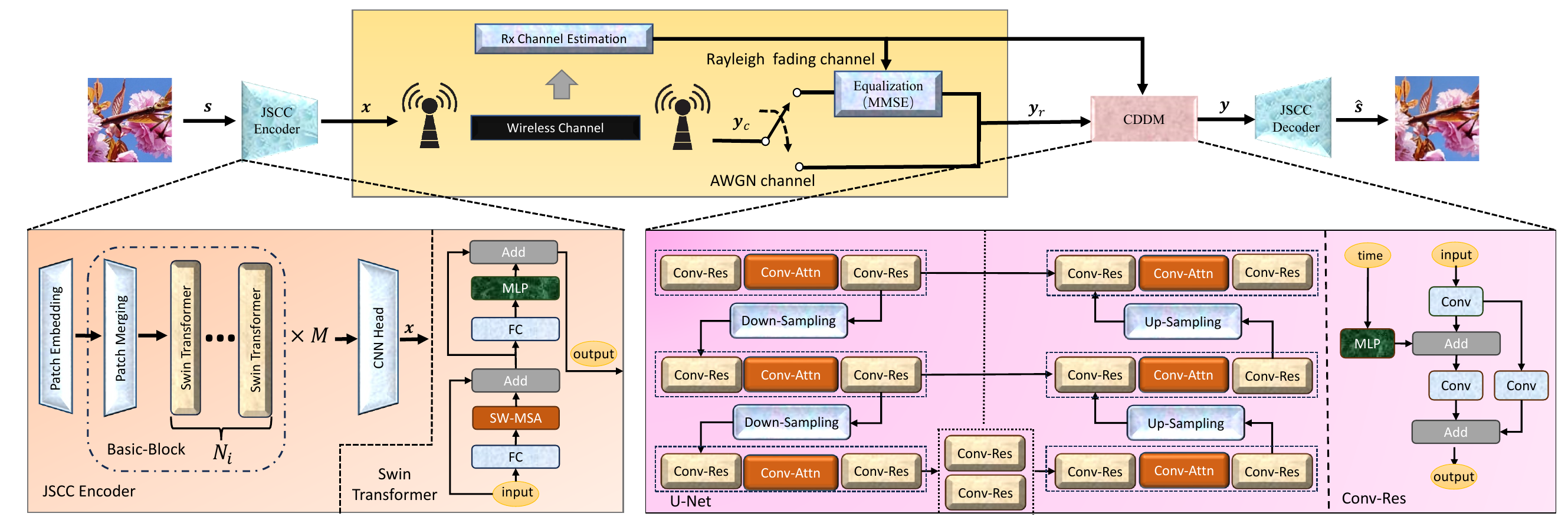}
  \end{center}
    \caption{{Architecture of the joint CDDM and JSCC system.}}
    \label{CDESC}
\end{figure*}

The contributions of this paper can be summarized as follows.
\begin{itemize}
  \item We design a CDDM module based on the U-Net framework in wireless communications, which lies after the channel equalization (or without channel equalization) over the Rayleigh channel (or AWGN channel). The CDDM module learns the distribution of the channel input signal to predict the channel noise and remove it. The model is trained through the forward diffusion process specially designed to adapt the channel models, requiring no knowledge of the current channel state. After training, the CDDM addresses the received signal after equalization with the corresponding sampling algorithm, succeeding in eliminating the channel noise.
  \item We derive the explicit condition probability of the received signal after equalization according to the channel mathmatical model and the equalization algorithm, which instructs us to design the corresponding forward diffusion process to match the conditional distribution. The training of the proposed CDDM is accomplished by maximizing the variational lower bound of the logarithm maximun likelihood function, which is relaxed by introducing a series of latent variables in the forward diffusion process. Furthermore, we decompose the variational lower bound into multiple components associated with the latent variables and derive the final loss function using re-parameterization and re-weighted techniques to optimize these components respectively. By utilizing the Bayesian conditional posterior probability, we obtain a sampling algorithm that successfully and effectively mitigates the channel noise. 
  \item We derive the sufficient condition for the reverse sampling algorithm reducing the conditional entropy of the received signal. Through Monte Carlo experiments, we discover the magnitude of the reduction in the upper bound of the conditional entropy differs from various sampling steps, providing insights for selecting the maximum sampling steps.
  \item We apply the CDDM to a semantic communications system based on the JSCC technique for wireless image transmission, called the joint CDDM and JSCC system. Experiments on the mean square error (MSE) between the transmitted signal and the received signal after CDDM prove that compared to the system without CDDM, the system with CDDM has a smaller MSE performance for both Rayleigh fading channel and AWGN channel, indicating that the proposed CDDM can effectively reduce the impact of channel noise through learning. Finally, extensive experimental results on different datasets demonstrate that the joint CDDM and JSCC system outperforms both the JSCC system and the traditional JPEG2000 with LDPC codec system under both AWGN and Rayleigh fading channels in terms of the peak signal-to-noise ratio (PSNR) and MSSSIM. We also evaluate its inherent robustness to channel estimation errors and its adaptability to various SNRs.
\end{itemize}

The rest of this paper is organized as follows. The system model is introduced in Section II. The detail of the proposed CDDM is presented in Section III. The joint CDDM and JSCC system for semantic communications is introduced in Section IV. Finally, extensive experimental results are presented in Section V, and conclusions are drawn in Section VI.


\section{System model}\label{I}
In this section, we describe the system which the proposed CDDM is employed after the channel equalization as shown in Fig. \ref{CDESC}. CDDM is trained using a specialized noise schedule adapted to the wireless channel, which enables it to effectively eliminate channel noise through sampling algorithm.

Let $\mathbf{x}\in \mathbb{R}^{2k}$ be the real-valued symbols. Here, $k$ is the number of channel uses. $\mathbf{x_c}\in\mathbb{C}^{k}$ is the complex-valued symbols which can be transmitted through the wireless channel, and the $i$-th transmitted symbol of $\mathbf{x_c}$ can be expressed as ${x_{c,i}}={x_i}+jx_{i+k}$, for $i=1,...,k.$

Thus, the $i$-th received symbol of the received signal $\mathbf{y_c}$ is
\begin{align}\label{receive signal}
y_{c,i}=h_{c,i}x_{c,i}+n_{c,i},
\end{align}
where $h_{c,i}\sim \mathbb{CN}(0,1)$ are independent and identically distributed (i.i.d.) Rayleigh fading gains, $\mathbf{x_{c}}$ has a power constraint $\mathbb{E}[||\mathbf{x_{c}}||^2_2]\leq 1$, and $n_{c,i}\sim \mathbb{CN} (0,2\sigma^2)$ are i.i.d. AWGN samples.


$\mathbf{y_c}$ is then addressed by equalization as $\mathbf{y_{eq}}\in \mathbb{C}^{k}$, following a normalization-reshape module outputing a real vector $\mathbf{y_r}\in\mathbb{R}^{2k}$. We consider that the receiver can obtain the channel state $\mathbf{h_c}=[h_{c,1},...,h_{c,k}]$ through channel estimation and in this paper, we apply minimum mean square error (MMSE) as the equalizer. Therefore, we can derive the conditional distribution of $\mathbf{y_r}$ with known $\mathbf{x}$ and $\mathbf{h_c}$, which can be formulated to instruct the forward diffusion and reverse sampling processes of CDDM.

\begin{proposition}\label{thm1}
With MMSE, the conditional distribution of $\mathbf{y_r}$ with known $\mathbf{x}$ and $\mathbf{h_c}$ under Rayleigh fading channel is
  \begin{gather}\label{zfyr}
    p(\mathbf{y_r}|\mathbf{x},\mathbf{h_c})\sim \mathcal{N}(\mathbf{y_r};\frac{1}{\sqrt{1+\sigma ^{2}}}\mathbf{W_s}\mathbf{x},\frac{\sigma ^{2}}{{1+\sigma ^{2}}}\mathbf{W}^2_{\mathbf{n}}),
  \end{gather}
where $\mathbf{H_r}=diag({\mathbf{h_r}})$, $\mathbf{h}_{\mathbf{r}}=\begin{bmatrix}|\mathbf{h_c}|\\|\mathbf{h_c}|\end{bmatrix}\in \mathbb{R}^{2k}$, and
 \begin{gather}\label{WsMMSE}
   \mathbf{W_s}=\mathbf{H}^2_{\mathbf{r}}(\mathbf{H}^2_\mathbf{r}+2\sigma^2\mathbf{I})^{-1},\mathbf{W_n}=\mathbf{H_r}(\mathbf{H}^2_\mathbf{r}+2\sigma^2\mathbf{I})^{-1}.
 \end{gather}
\end{proposition}
\begin{proof}
Based on the defination, $\mathbf{W_s}$ and $\mathbf{W_n}$ are diagonal matrix, where the $i$-th and ($i+k$)-th diagonal element are
  \begin{align}\label{Wselement}
    {W_{s,i}}= {W_{s,i+k}}= \frac{|h_{c,i}|^2}{|h_{c,i}|^2+2\sigma^2}\nonumber,\\
    {W_{n,i}}= {W_{n,i+k}}= \frac{|h_{c,i}|}{|h_{c,i}|^2+2\sigma^2}.
  \end{align}

The $i$-th output of MMSE ${y_{eq,i}}$ can be expressed as
  \begin{align}\label{zfeq}
    {y_{eq,i}}=\frac{|h_{c,i}|^2x_{c,i}+h_{c,i}^Hn_{c,i}}{|h_{c,i}|^2+2\sigma^2}.
  \end{align}

Based on (\ref{Wselement}), we have
  \begin{align}
    \frac{|h_{c,i}|^2x_{c,i}}{|h_{c,i}|^2+2\sigma^2}={W_{s,i}}x_{c,i}.
  \end{align}

With the resampling trick, the conditional distributions of real part and imaginary part of $\frac{h_{c,i}^Hn_{c,i}}{|h_{c,i}|^2+2\sigma^2}$ are
  \begin{align}\label{renoise}
    p(Re(\frac{h_{c,i}^Hn_{c,i}}{|h_{c,i}|^2+2\sigma^2})|h_{c,i}) &\sim\mathcal{N}(0,{\sigma^{2}}(\frac{|h_{c,i}|}{|h_{c,i}|^2+2\sigma^{2}})^2)\nonumber\\
&=\mathcal{N}(0,{\sigma^{2}}W_{n,i}^2),
  \end{align}
  \begin{align}\label{imnoise}
    p(Im(\frac{h_{c,i}^Hn_{c,i}}{|h_{c,i}|^2+2\sigma^2})|h_{c,i})\sim \mathcal{N}(0,{\sigma^{2}}W_{n,i}^2).
  \end{align}
 
Accordingly, we can rewrite $\mathbf{y_r}$ as
  \begin{align}
    \mathbf{y_r}=\frac{1}{\sqrt{1+\sigma^{2}}}(\mathbf{W_sx}+\mathbf{n_r}),
  \end{align}
and the distribution $p(\mathbf{n_r}|\mathbf{h_c})$ is $\mathcal{N}(0,\sigma^{2}\mathbf{W}^2_{\mathbf{n}})$.

Therefore, we have 
\begin{gather}\label{zfyr_final}
  p(\mathbf{y_r}|\mathbf{x},\mathbf{h_c})\sim \mathcal{N}(\mathbf{y_r};\frac{1}{\sqrt{1+\sigma ^{2}}}\mathbf{W_s}\mathbf{x},\frac{\sigma ^{2}}{{1+\sigma ^{2}}}\mathbf{W}^2_{\mathbf{n}}).
\end{gather}
\end{proof}

Similarly, we have the following proposition for AWGN channel.
\begin{proposition}\label{thm2}
  Under AWGN channel, the conditional distribution of $\mathbf{y_r}$ with known $\mathbf{x}$ is
  \begin{align}
    p(\mathbf{y_r}|\mathbf{x})\sim \mathcal{N}(\mathbf{y_r};\frac{1}{\sqrt{1+\sigma ^{2}}}\mathbf{W_s}\mathbf{x},\frac{\sigma ^{2}}{{1+\sigma ^{2}}}\mathbf{W}^2_{\mathbf{n}}),
  \end{align}
where $\mathbf{W_s}$ and $\mathbf{W_n}$ both bacome $\mathbf{I}_{2k}$ under AWGN channel.

\end{proposition}

Proposition 1 an Proposition 2 demonstrate that the channel noise after equalization and normalization-reshape can be re-sampled using $\mathbf{\epsilon} \sim \mathcal{N}(0,\mathbf{I}_{2k})$. Additionally, the noise coefficient matrix $\mathbf{W_n}$ is related to the modulo form of $\mathbf{h_c}$. As a result, $\mathbf{y_r}$ can be expressed as
\begin{figure*}[t]
  \begin{center}
    \includegraphics*[width=18cm]{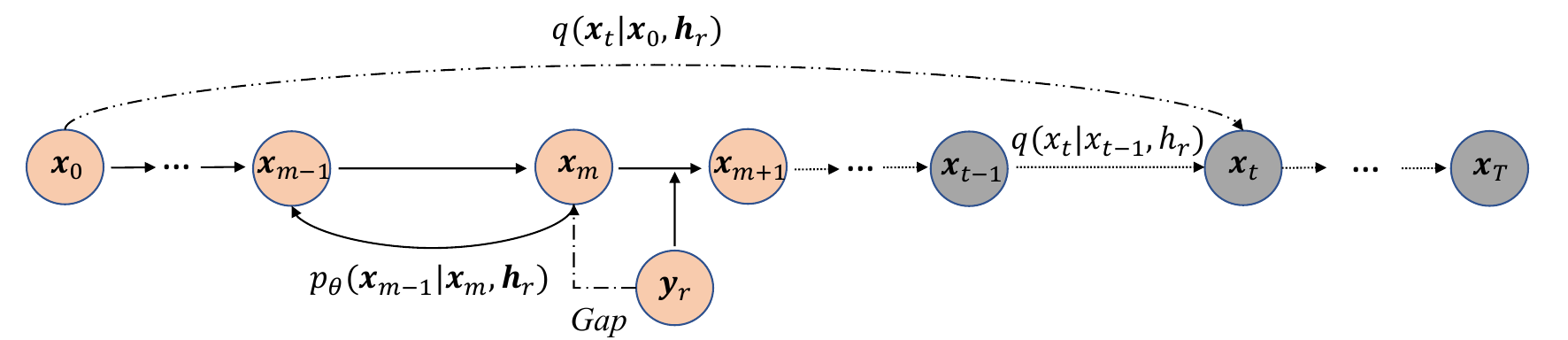}
  \end{center}
    \caption{{The forward diffusion process and reverse sampling process of the proposed CDDM.}}
    \label{ourstep}
\end{figure*}

\begin{align}
  \mathbf{y_r}=\frac{1}{\sqrt{1+\sigma ^{2}}}\mathbf{W_s}\mathbf{x}+\frac{\sigma}{\sqrt{1+\sigma ^{2}}}\mathbf{W_n}\epsilon.
\end{align}

Therefore, the proposed CDDM is trained to obtain $\mathbf{\epsilon_{\theta}}(\cdot )$, which is an estimation of $\mathbf{\epsilon}$. Here, $\mathbf{\theta}$ is all parameters of CDDM. By using $\mathbf{\epsilon_\theta}(\cdot)$ and $\mathbf{W_n}$, a sampling algorithm is proposed to obtain $\mathbf{y}$ with the aim to recover $\mathbf{W_sx}$, which will be described in the next section. 

\section{Channel Denoising Diffusion Models}
The whole strcuture of the CDDM forward diffusion and reverse sampling process is illustrated in Fig. \ref{ourstep}. In this section, we first describe the training algorithm and sampling algorithm of the proposed CDDM. We then derive the sufficient condition for the reverse sampling algorithm reducing the conditional entropy of the received signal. 

\subsection{Training Algorithm of CDDM}
For the forward process of the proposed CDDM, the original source $\mathbf{x}_0$ is
\begin{align}
  \mathbf{x}_0=\mathbf{W_sx}.
\end{align}
Let $T$ be the hyperparameter. Similar to (1), for all $t\in \{1,2,...,T\}$, we define
\begin{align}\label{x_t}
  \mathbf{x}_t=\sqrt{\alpha _t}\mathbf{x}_{t-1}+\sqrt{1-\alpha _t}\mathbf{W_n}\mathbf{\epsilon},
\end{align}
and then it can be re-parametered as
\begin{align}\label{reforward}
  \mathbf{x}_t=\sqrt{\bar{\alpha}_t}\mathbf{x}_0+\sqrt{1-\bar{\alpha}_t}\mathbf{W_n}\mathbf{\epsilon},
\end{align}
such that the distribution $q(\mathbf{x}_t|\mathbf{x}_0,\mathbf{h_r})$ is
\begin{gather}\label{distributionforward}
  {q(\mathbf{x}_t|\mathbf{x}_0,\mathbf{h_r})\sim \mathcal{N}(\mathbf{x}_t;\sqrt{\bar{\alpha}_t}\mathbf{x}_0,({1-\bar{\alpha}_t})\mathbf{W}^2_{\mathbf{n}})}.
\end{gather}

Based on (\ref{zfyr}) and (\ref{distributionforward}), if $\bar{\alpha}_m=\frac{1}{1+\sigma^{2}}$, the Kullback-Leibler (KL) divergence is
\begin{align}\label{KLzero}
  D_{KL}(q(\mathbf{x}_m|\mathbf{x}_0,\mathbf{h_r})||p(\mathbf{y_r}|\mathbf{x}_0,\mathbf{h_c}))=0,
\end{align}

This indicates that through defining a forward diffusion process, we progressively generate a signal following the same distribution as the one passed through the real channel and equalizer.
Such that \textbf{CDDM can be trained on $\mathbf{x}_m$ instead of $\mathbf{y_r}$}. $\mathbf{x}_m$ is defined by $m$ steps as (\ref{x_t}) such 
that in sampling process, the predicted distribution by CDDM can be decomposed into $m$ small steps and each of them is $p_{\mathbf{\theta}}(\mathbf{x}_{t-1}|\mathbf{x}_t,\mathbf{h_r})$ for $t\in \{1,2,...,m\}$.

The goal of CDDM is to recover $\mathbf{x}_0$ by learning the distribution of $\mathbf{x}_0$ and removing the channel noise. Therefore, the training of CDDM is performed by optimizing the variational bound on negative log likehood $L$. The variational bound of $L$ is form by $\mathbf{x}_{0:m}$ and $\mathbf{y_r}$, which is given by
\begin{align}\label{var-bound}
  &L=\mathbb{E}\ [-\log\ p_{\mathbf{\theta}}(\mathbf{x}_0|\mathbf{h_r})]\le \mathbb{E}_q[-\log(\frac{p_{\mathbf{\theta}}(\mathbf{x}_{0:m},\mathbf{y_r}|\mathbf{h_r})}{q(\mathbf{x}_{1:m},\mathbf{y_r}|\mathbf{x}_0,\mathbf{h_r})})]\nonumber\\
  &=\mathbb{E}_q\ \underbrace{[D_{KL}(q(\mathbf{y_r}|\mathbf{x}_0,\mathbf{h_r})||p(\mathbf{y_r}|\mathbf{h_r}))}_{L_y}-\underbrace{\log p_{\mathbf{\theta}}(\mathbf{x}_0|\mathbf{x}_1,\mathbf{h_r})}_{L_0}\nonumber\\
  &+\underbrace{D_{KL}(q(\mathbf{x}_m|\mathbf{y_r},\mathbf{x}_0,\mathbf{h_r})||p_{\mathbf{\theta}}(\mathbf{x}_m|\mathbf{y_r},\mathbf{h_r}))}_{L_m}\nonumber\\
  &+{\sum_{t=1}^{m}\underbrace{D_{KL}(q(\mathbf{x}_{t-1}|\mathbf{x}_t,\mathbf{x}_0,\mathbf{h_r})||p_{\mathbf{\theta}}(\mathbf{x}_{t-1}|\mathbf{x}_t,\mathbf{h_r}))}_{L_{t-1}}}],
\end{align}
where $L_m$ instructs to select the hyperparameter $m$. In this paper, we select $m$ by
\begin{align}\label{selectm}
  arg\min_{m}\   2\sigma^{2}-\frac{1-\bar{\alpha}_m}{\bar{\alpha}_m}.
\end{align}

Similar to the process in \cite{Ho}, $L_{t-1}$  can be calculated in closed-form using the Rao-Blackwellized method. The optimization object of $L_{t-1}$ can be simplified by adopting re-parameterization and re-weighting methods as following
\begin{align}\label{simL}
  \mathbb{E}_{\mathbf{x}_0,\mathbf{\epsilon}}(||\mathbf{W_n\epsilon}-\mathbf{W_n\epsilon_{\theta}}(\mathbf{x}_t,\mathbf{h_r},t)||^2_2),
\end{align}
where $\mathbf{\epsilon_{\theta}}(\mathbf{x}_t,\mathbf{h_r},t)$ is the output of CDDM. Moreover, (\ref{simL}) can be re-weighted by ignoring the noise coefficient matrix $\mathbf{W_n}$ as following
\begin{align}\label{L_last}
\mathbb{E}_{\mathbf{x}_0,\mathbf{\epsilon}}(||\mathbf{\epsilon}-\mathbf{\epsilon_{\theta}}(\sqrt{\bar{\alpha}_t}\mathbf{x}_0+\sqrt{1-\bar{\alpha}_t}\mathbf{W_n}\mathbf{\epsilon})||^2_2).
\end{align}

Finally, to optimize (\ref{L_last}) for all $t\in\{1,2,...,T\}$, the loss function of the proposed CDDM is expressed as follows
\begin{align}\label{L_CDDM}
  L_{CDDM}(\mathbf{\theta})=\mathbb{E}_{\mathbf{x}_0,\mathbf{\epsilon},t}(||\mathbf{\epsilon}-\mathbf{\epsilon_{\theta}}(\sqrt{\bar{\alpha}_t}\mathbf{x}_0+\sqrt{1-\bar{\alpha}_t}\mathbf{W_n}\mathbf{\epsilon})||^2_2).
\end{align}

In summary, the proposed CDDM has the capability to estimate noise, due to its ability to learn to approximate the real posterior distribution $q(\mathbf{x}_{t-1}|\mathbf{x}_t,\mathbf{x}_0,\mathbf{h_r})$ with its parameterized distribution $p_{\mathbf{\theta}}(\mathbf{x}_{t-1}|\mathbf{x}_t,\mathbf{h_r})$ in the training process. The distribution approximation can be derived into noise estimation, as shown in (\ref{L_CDDM}). The training procedures of the proposed CDDM are summarized in Algorithm \ref{trainCDDM}.
\begin{algorithm}[t]
  \hspace*{0.02in} {\bf \small{Input:}}
	\small{Training set $S$, hyper-parameter $T$ and $\bar{\alpha}_t$.} \\
	\hspace*{0.02in} {\bf \small{Output:}}
	\small{The trained CDDM.}
	\caption{Training algorithm of CDDM}
	\label{trainCDDM}
	\begin{algorithmic}[1] 
    \WHILE {the training stop condition is not met}
    \STATE Randomly sample $\mathbf{x}$ from $S$
    \STATE Randomly sample $t$ from $Uniform(\{1,...,T\})$
    \STATE Sapmle $|\mathbf{h_c}|$ and compute $\mathbf{H_r}$, $\mathbf{W_s}$ and $\mathbf{W_n}$
    \STATE Randomly sample $\mathbf{\epsilon}$ from $\mathcal{N}(0,\mathbf{I}_{2k})$
    \STATE Take gradient descent step according to (\ref{x_t}) and (\ref{L_CDDM})\\
    $\nabla_\mathbf{\theta}(||\mathbf{\epsilon}-\mathbf{\epsilon _\theta}(\sqrt{\bar{\alpha}_t}\mathbf{W_sx}+\sqrt{1-\bar{\alpha}_t}\mathbf{W_n}\mathbf{\epsilon})||^2_2)$
    \ENDWHILE
	\end{algorithmic}
\end{algorithm}

\subsection{Sampling Algorithm of CDDM}\label{B}
To reduce the time consumption of sampling process and recover the transmitted signal accurately, (\ref{var-bound}) implies that selecting $m$ according to (\ref{selectm}) and setting $\mathbf{x}_m=\mathbf{y_r}$ is a promising way.
By utilizing the received signal $\mathbf{y_r}$, only $m$ steps are needed to be excuted. For each time step $t\in\{1,2,...,m\}$, the trained CDDM outputs $\mathbf{\epsilon_\theta}(\mathbf{x}_t,\mathbf{h_r},t)$, which attempts to predict $\mathbf{\epsilon}$ from $\mathbf{x}_t$ without knowledge of $\mathbf{x}_0$. A sampling algorithm is required to sample $\mathbf{x}_{t-1}$. The process is excuted for $m$ times such that $\mathbf{x}_0$ can be computed out finally.

We first define the sampling process $f(\mathbf{x}_{t-1})$ with the knowledge of $\mathbf{\epsilon}$ as following
\begin{align}
  f(\mathbf{x}_{t-1})=q(\mathbf{x}_{t-1}|\mathbf{x}_t,\mathbf{x}_0,\mathbf{h_r}).
\end{align}

Applying Bayes rule, the distribution can be expressed as a Gaussian distribution
\begin{gather}
  q(\mathbf{x}_{t-1}|\mathbf{x}_t,\mathbf{x}_0,\mathbf{h_r})\nonumber\\
  \sim \mathcal{N}(\mathbf{x}_{t-1};\sqrt{\bar{\alpha}_{t-1}}\mathbf{x}_0+\sqrt{1-\bar{\alpha}_{t-1}}\frac{\mathbf{x}_t-\sqrt{\bar{\alpha}_t}\mathbf{x}_0}{\sqrt{1-\bar{\alpha}_t}},0),
\end{gather}
\begin{algorithm}[t]
  \hspace*{0.02in} {\bf \small{Input:}}
  \small{$\mathbf{y_r}$,$\mathbf{h_r}$,hyperparameter $m$} \\
  \hspace*{0.02in} {\bf \small{Output:}}
  \small{$\mathbf{y}$}
  \caption{Sampling algorithm of CDDM}
  \label{sampleCDDM}
  \begin{algorithmic}[1]
    \STATE $\mathbf{x}_m=\mathbf{y_r}$ 
    \FOR {$t=m,...,2$}
    \STATE $\mathbf{z}=\mathbf{W_n}\mathbf{\epsilon_\theta}(\mathbf{x}_t,\mathbf{h_r},t)$
    \STATE $\mathbf{x}_{t-1}=\sqrt{\bar{\alpha}_{t-1}}(\frac{\mathbf{x}_t-\sqrt{1-\bar{\alpha}_t}\mathbf{z}}{\sqrt{\bar{\alpha}_t}})+\sqrt{1-\bar{\alpha}_{t-1}}\mathbf{z}$
    \ENDFOR
    \STATE $t=1$
\STATE $\mathbf{z}=\mathbf{W_n}\mathbf{\epsilon_\theta}(\mathbf{x}_1,\mathbf{h_r},1)$
    \STATE $\mathbf{y}=\frac{\mathbf{x}_1-\sqrt{1-\bar{\alpha}_{1}}\mathbf{z}}{\sqrt{\bar{\alpha}_1}}$
  \end{algorithmic}
\end{algorithm}
where $\mathbf{x}_0$ is acquired by re-writing (\ref{reforward}) as following
\begin{align}
  \mathbf{x}_0=\frac{1}{\sqrt{\bar{\alpha}_t}}(\mathbf{x}_t-\sqrt{1-\bar{\alpha}_t}\mathbf{W_n}\mathbf{\epsilon}).
\end{align}

However, only $\mathbf{\epsilon_\theta}(\mathbf{x}_t,\mathbf{h_r},t)$ is available for sampling. $\mathbf{x}_0$ is derived through an estimation process by replacing $\mathbf{\epsilon}$ with $\mathbf{\epsilon_\theta}(\mathbf{x}_t,\mathbf{h_r},t)$ as following
\begin{align}\label{predictx-0}
  {\hat{\mathbf{x}}_0}=\frac{1}{\sqrt{\bar{\alpha}_t}}(\mathbf{x}_t-\sqrt{1-\bar{\alpha}_t}\mathbf{W_n}\mathbf{\epsilon_\theta}(\mathbf{x}_t,\mathbf{h_r},t)).
\end{align}

As a result, the sampling process is replaced with
\begin{align}
  f_\mathbf{\theta}(\mathbf{x}_{t-1})=p_\mathbf{\theta}(\mathbf{x}_{t-1}|\mathbf{x}_t,\hat{\mathbf{x}}_0,\mathbf{h_r}).
\end{align}

Without the knowledge of $\mathbf{\epsilon}$, a sample of $\mathbf{x}_{t-1}$ is
\begin{align}\label{samplext}
  \mathbf{x}_{t-1}=&\sqrt{\bar{\alpha}_{t-1}}\underbrace{(\frac{1}{\sqrt{\bar{\alpha}_t}}(\mathbf{x}_t-\sqrt{1-\bar{\alpha}_t}\mathbf{W_n}\mathbf{\epsilon_\theta}(\mathbf{x}_t,\mathbf{h_r},t)))}_{estimate\ \mathbf{x}_0}\nonumber\\
  &+\underbrace{\sqrt{1-\bar{\alpha}_{t-1}}\mathbf{W_n\epsilon_\theta}(\mathbf{x}_t,\mathbf{h_r},t)}_{sample\ \mathbf{x}_{t-1}}.
\end{align}

Note that for the last step $t=1$, we only predict $\mathbf{x_0}$ such that sampling is taken as
\begin{align}
  \mathbf{y}=\frac{1}{\sqrt{\bar{\alpha}_1}}(\mathbf{x}_1-\sqrt{1-\bar{\alpha}_1}\mathbf{W_n\epsilon_\theta}(\mathbf{x}_1,\mathbf{h_r},1)).
\end{align}
The sampling method is summarized in Algorithm \ref{sampleCDDM}.

\subsection{Analysis on the conditional entropy}
To explain the denoising ability of the CDDM, we compare the conditional entropy between $\mathbf{x}_t$ and $\mathbf{x}_{t-1}$, where $\mathbf{x}_t$ is considered as the receiving signal because
(\ref{KLzero}) has shown that $\mathbf{x}_t$ can belong to the same conditional distribution as the received signal.

For all $t\in\{1,2,...,T\}$, $\mathbf{x}_t$ is acquired as (\ref{reforward}). According to (\ref{distributionforward}), we can get the conditional entropy of the $i$-th element of $\mathbf{x}_t$ as $\mathcal{H}({x}_{t,i}|\mathbf{x}_0,\mathbf{h})=\frac{1}{2}\ln({W^2_{n,i}(1-\bar{\alpha}_t)})+C$, $i=1,2,...,2k$. Here, $C$ is a constant.
$\mathbf{x}_{t-1}$ is sampled as (\ref{samplext}).
However, $\mathbf{x}_t$ is unknown in $\mathcal{H}({x}_{t-1,i}|\mathbf{x}_0,\mathbf{h})$. We can reparameter (\ref{samplext}) with (\ref{reforward}) and obtain
\begin{align}\label{entropyxt_1}
  \mathbf{x}_{t-1}=\sqrt{\bar{\alpha}_{t-1}}\mathbf{x}_0+\beta_t\mathbf{W_n\epsilon}-\beta_t\mathbf{W_n\epsilon_\theta}(\cdot)+\gamma_{t-1} \mathbf{W_n\epsilon_\theta}(\cdot),
\end{align}
where $\beta_t=\frac{\sqrt{1-\bar{\alpha}_t}}{\sqrt{\alpha_t}}$ and $\gamma_t=\sqrt{1-\bar{\alpha}_{t}}$.
$\mathbf{\epsilon}\sim \mathcal{N}(0,\mathbf{I})$ and thus $\mathbf{x}_{t-1}$ is a random variable with respect to $\mathbf{\epsilon}$ with unknown distribution.

Now, we introduce two assumptions for the following analysis.
\begin{figure}[t]
  \begin{center}
    \includegraphics*[width=8.65cm]{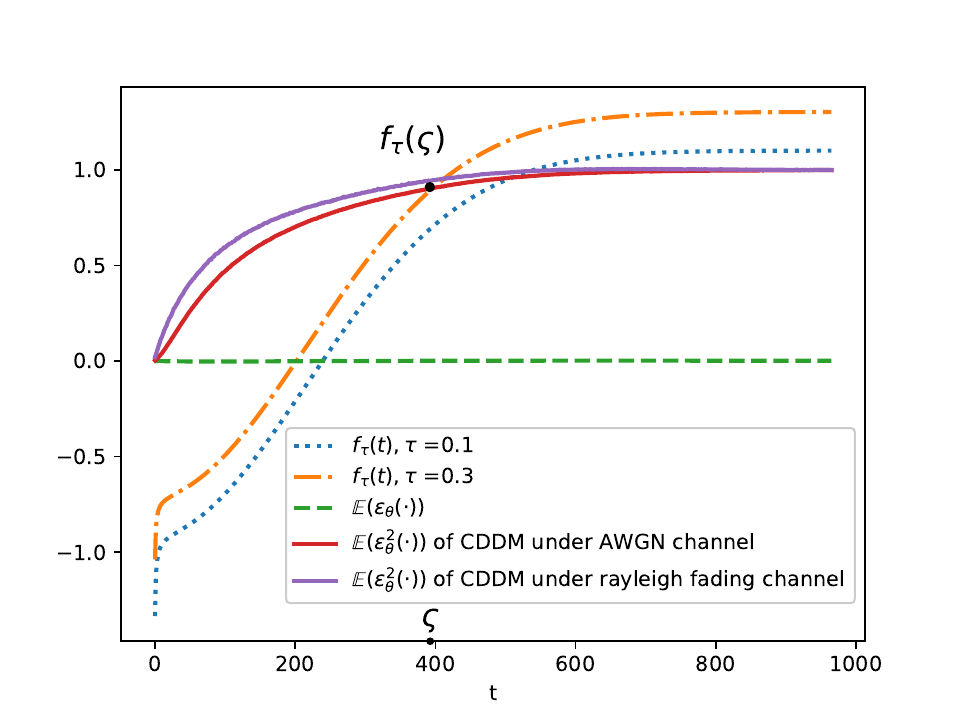}
  \end{center}
    \caption{\small{Experiment results of $\mathbb{E}(\epsilon_\theta(\cdot))$ and $\mathbb{E}(\epsilon^2_\theta(\cdot))$ with theoretical values of $f_\tau(t)$ versus sampling step $t$. The black dot marked the maximum sampling step, below which the model satisfies the sufficient condition under AWGN channel.}}
    \label{prove}
\end{figure}

\begin{assumption}
  There exists a constant bound $\tau>0$ on the element-wise loss function:
\begin{align}\label{assumption1}
  \mathbb{E}_{\mathbf{\epsilon}}(||{\epsilon_i}-{\epsilon_{\theta,i}}(\cdot)||^2_2)\leq \tau.
\end{align}
\end{assumption}
This reasonable and necessary assumption is derived from the fact that the network is optimized sufficiently, meaning the loss function $\mathbb{E}_{\mathbf{\epsilon}}(||\mathbf{\epsilon}-\mathbf{\epsilon_{\theta}}(\cdot)||^2_2)\leq \chi $, which can be written into element-wise form as (\ref{assumption1}).

\begin{assumption}
  The mathematical expectation of network output is 0, i.e.,
\begin{align}
  \mathbb{E}_{\mathbf{\epsilon}}({\epsilon_{\theta,i}}(\cdot))=0.
\end{align}
\end{assumption}
This assumption will be verified through Monte-Carlo in the following. Thus, we have the following theorem.
\begin{theorem}
Based on the two assumptions mentioned above, for all $t\in\{1,2,...,T\}$ and $i=1,2,...,2k$, the sufficiency condition of
  \begin{align}
    \mathcal{H}({x}_{t-1,i}|\mathbf{x}_0,\mathbf{h})\leq \mathcal{H}({x}_{t,i}|\mathbf{x}_0,\mathbf{h})
    \end{align}
    is
  \begin{align}
  \mathbb{E}_{\mathbf{\epsilon}}(\epsilon^2_{\theta,i}(\cdot))\geq \frac{1-\bar{\alpha}_t-\beta_t\gamma_{t-1}}{\gamma^2_{t-1}-\beta_t\gamma_{t-1}}-\frac{\beta^2_t-\beta_t\gamma_{t-1}}{\gamma^2_{t-1}-\beta_t\gamma_{t-1}}\tau.
  \end{align}
\end{theorem}

\begin{IEEEproof}
According to Assumption 1, we can derive the cross-correlation coefficient of the two random variables ${\epsilon_i}$ and ${\epsilon_{\theta,i}}(\cdot)$ as following
\begin{gather}
  \mathbb{E}_{\mathbf{\epsilon}}(||{\epsilon_i}-{\epsilon_{\theta,i}}(\cdot)||^2_2)=\mathbb{E}({\epsilon_i^2}-2{\epsilon_i}{\epsilon_{\theta,i}}(\cdot)+{\epsilon_{\theta,i}^2}(\cdot)) \leq {\tau}.
\end{gather}
We then have
\begin{gather}\label{optimize}
  2\mathbb{E}({\epsilon_i}{\epsilon_{\theta,i}}(\cdot))\geq  1-\mathbf{\tau}+\mathbb{E}({\epsilon_{\theta,i}^2}(\cdot)).
\end{gather}

Let ${\pi}^2_{t-1,i}$ be the variance of ${x}_{t-1,i}$. According to (\ref{entropyxt_1}), (\ref{optimize}) and Assumption 2, we have
\begin{align}\label{cor}
&{\pi}^2_{t-1,i}=\mathbb{E}({x}_{t-1,i}^2)-\mathbb{E}^2({x}_{t-1,i})\nonumber\\
  &={W_{n,i}^2}\mathbb{E}(\beta^2_t{\epsilon_i^2}+(\beta_t-\gamma_{t-1})^2{\epsilon_{\theta,i}^2}(\cdot){\setlength\arraycolsep{0.3pt}-}2\beta_t(\beta_t{\setlength\arraycolsep{0.3pt}-}\gamma_{t-1}){\epsilon_i}{\epsilon_{\theta,i}}(\cdot))\nonumber\\
 & \leq W_{n,i}^2(\beta_t^2+(\beta_t-\gamma_{t-1})^2\mathbb{E}(\epsilon_{\theta,i}^2(\cdot))\nonumber\\
 & -\beta_t(\beta_t-\gamma_{t-1})(1-\tau+\mathbb{E}(\epsilon_{\theta,i}^2)))\nonumber\\
 & = W_{n,i}^2((\gamma^2_{t-1}-\beta_t \gamma_{t-1})\mathbb{E}(\epsilon^2_{\theta,i}(\cdot))\nonumber\\
 & +\beta_t\gamma_{t-1}+(\beta_t^2-\beta_t\gamma_{t-1})\tau).
\end{align}

\begin{figure}[t]
  \begin{center}
    \includegraphics*[width=8.5cm]{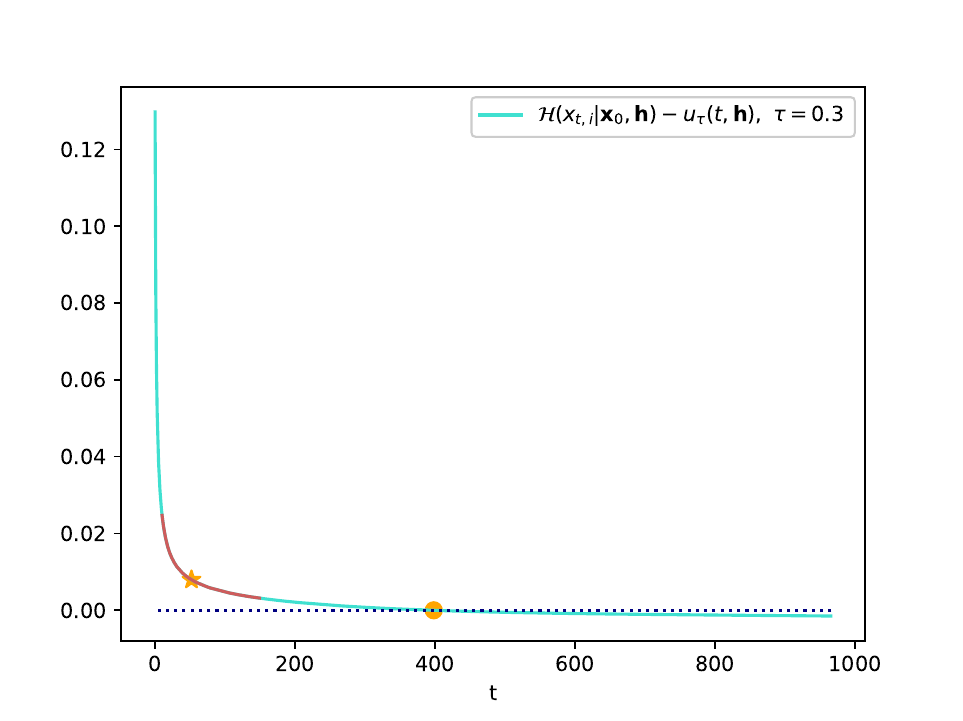}
  \end{center}
    \caption{\small{Experiment results of $\mathcal{H}(x_{t,i}|\mathbf{x}_0,\mathbf{h})-u_\tau(t,\mathbf{h})$ at $\tau=0.3.$}}
    \label{recommend}
\end{figure}

Let $u_\tau(t,\mathbf{h})$ be the upper bound of $\mathcal{H}({x}_{t-1,i}|\mathbf{x}_0,\mathbf{h})$. With the maximum entropy principle, we have
\begin{align}\label{upperbound}
  \mathcal{H}({x}_{t-1,i}|\mathbf{x}_0,\mathbf{h}) &\leq \frac{1}{2}\ln(\pi^2_{t-1,i})+C\nonumber\\
   &\leq\frac{1}{2}\ln(W_{n,i}^2((\gamma^2_{t-1}-\beta_t \gamma_{t-1})\mathbb{E}(\epsilon^2_{\theta,i}(\cdot))\nonumber\\
   &+\beta_t\gamma_{t-1}+(\beta_t^2-\beta_t\gamma_{t-1})\tau))+C \nonumber\\
   &\triangleq u_\tau(t,\mathbf{h}).
\end{align}
Here, we have $\gamma^2_{t-1}-\beta_t \gamma_{t-1}<0$. Therefore, it is easy to obtain the necessity and sufficiency conditional for the inequalities $u_\tau(t,\mathbf{h})\leq \mathcal{H}({x}_{t,i}|\mathbf{x}_0,\mathbf{h})$ as following
\begin{gather}\label{suff_condition}
  \mathbb{E}(\epsilon^2_{\theta,i}(\cdot))\geq \frac{1-\bar{\alpha}_t-\beta_t\gamma_{t-1}}{\gamma^2_{t-1}-\beta_t\gamma_{t-1}}-\frac{\beta^2_t-\beta_t\gamma_{t-1}}{\gamma^2_{t-1}-\beta_t\gamma_{t-1}}\tau \triangleq f_\tau(t).
\end{gather}

Taking the necessity and sufficiency condition into (\ref{upperbound}), we can get the sufficiency condition as the theory. 
\end{IEEEproof}

In Fig. \ref{prove}, the dashed line represents the Monte Carlo results of $\mathbb{E}(\epsilon_\theta(\cdot))$ approaching zero, which proves that Assumption 2 holds in the proposed model. It also demonstrates that there exists a limitation $\varsigma$. If $t\leq \varsigma $, the condition (\ref{suff_condition}) holds. This suggests that the number of sampling steps should be limited in order to achieve performance improvements. Fig. \ref{recommend} shows the value of $\mathcal{H}(x_{t,i}|\mathbf{x}_0,\mathbf{h})-u_\tau(t,\mathbf{h})$ at $\tau=0.3$ versus sampling step $t$. It is observed that the curve initially exhibits a sharp decline and subsequently levels off rapidly. Considering two figures together, the sampling step of CDDM can not be determines utilizing (\ref{selectm}) in case where the channel noise power is excessively high, as it would exceed the threshold $\varsigma $. Furthermore, even if the sampling step is below $\varsigma$, the gradient becomes very small when it falls within the flattened region. This can lead to the conditional entropy remaining stagnant, resulting in no performance improvement.
On the other hand, if the sampling step is too small, the channel noise may not be eliminated sufficiently. Based on the analysis above, we recommend to set the maximum sapmling step $t_{max}\in[10,150]$ as shown in the Fig. \ref{recommend} with red line. Correspondingly, (\ref{selectm}) is revised into
\begin{align}\label{revisem}
  m=\min (t_{max},\arg\min_{m}\   2\sigma^{2}-\frac{1-\bar{\alpha}_m}{\bar{\alpha}_m}).
\end{align}

\section{The joint CDDM and JSCC for semantic communications}
In this section, the proposed CDDM is applied into a semantic communications system based on JSCC for wireless image transmission.
\subsection{System Structure}
An overview architecture of the joint CDDM and JSCC system is shown in Fig. \ref{CDESC}. An RGB source image $\mathbf{s}$ is encoded as transmitted signal $\mathbf{x}\in \mathbb{R}^{2k}$ by a JSCC encoder. In this paper, the JSCC is built upon the Swin Transformer\cite{Ze} backbone, which has a more powerful expression ability than vision transformer by replacing the standard multi-head self attention in vision transformer with a shift window multi-head self attention.
$\mathbf{x}$ is then transmitted and processed into $\mathbf{y_r}$ at the receiver, as described in Section \ref{I}. At the receiver, the proposed CDDM removes the channal noise from $\mathbf{y_r}$ using Algorithm \ref{sampleCDDM}. Following this, the output of CDDM is fed into the JSCC decoder to reconstruct the source image $\mathbf{\hat{s}}$.
\subsection{Training algorithm}
\addtolength{\topmargin}{0.015in}
The entire training algorithm of the joint CDDM and JSCC system consists of three stages. In the first stage, the JSCC encoder and decoder are trained jointly through the channel shown in Fig. \ref{CDESC}, except for the CDDM module, to minimize the distance $d(\mathbf{s,\hat{s}})$. Therefore, the loss function for this stage is given by
\begin{align}
  L_1(\mathbf{\phi,\varphi})&=\mathbb{E}_{\mathbf{s}\sim p_\mathbf{s}}\mathbb{E}_{\mathbf{y_r}\sim p_{\mathbf{y_r|s}}}d(\mathbf{s},\mathbf{\hat{s}}).
\end{align}
where $\mathbf{\phi}$ and $\mathbf{\varphi}$ encapsulate all parameters of JSCC encoder and decoder respectively.

In the second stage, the parameters of the JSCC encoder are fixed such that CDDM can learn the distribution of $\mathbf{x}_0$ via Algorithm \ref{trainCDDM}. The training process is not affected by the channel noise power because Algorithm \ref{trainCDDM} has a special forward diffusion process, and the process has been designed specially to simulate the distribution of channel noise. Benefitting from this, CDDM is designed for handling various channel conditions and requires only one training process.

In the third stage, the JSCC decoder is re-trained jointly with the trained JSCC encoder and CDDM to minimize $d(\mathbf{s,\hat{s}})$. The entire joint CDDM and JSCC system is performed through the real channel, while only the parameters of the decoder are updated. The loss function is derived as
\begin{align}
  L_3(\mathbf{\varphi})=\mathbb{E}_{\mathbf{y}\sim p_{\mathbf{y|s}}}d(\mathbf{s},\mathbf{\hat{s}}).
\end{align}
The training algorithm is summarized in Algorithm \ref{trainCDESC}.
\begin{algorithm}[t]
  \hspace*{0.02in} {\bf{Input:}}
	\small{Training set $\mathbf{S}$, hyper-parameter $T$, $\bar{\alpha_t}$, and the channel estimation results $\mathbf{h_c}$ and $\sigma^2$.} \\
	\hspace*{0.02in} {\bf{Output:}}
	\small{The well-trained joint CDDM and JSCC system.}
	\caption{Training algorithm of the joint CDDM and JSCC}
	\label{trainCDESC}
	\begin{algorithmic}[1]
    \WHILE {the training stop condition of stage one is not met }
    \STATE Randomly sample $\mathbf{s}$ from $S$
    \STATE Perform forward propagation through channel without CDDM.
    \STATE Compute $L_1(\mathbf{\mathbf{\phi, \varphi}})$ and update $\mathbf{\phi, \varphi}$
    \ENDWHILE
    \WHILE {the training stop condition of stage two is not met }
    \STATE Randomly sample $\mathbf{s}$ from $S$
    \STATE Compute $\mathbf{s}$ as $\mathbf{x}$
    \STATE Train CDDM with Algorithm \ref{trainCDDM}.
    \ENDWHILE
    \WHILE {the training stop condition of stage three is not met }
    \STATE Randomly sample $\mathbf{s}$ from $S$
    \STATE Perform forward propagation through channel with noise power $\sigma^2$ with the trained CDDM
    \STATE Compute $L_3(\mathbf{\varphi})$ and update $\mathbf{\varphi}$
    \ENDWHILE
	\end{algorithmic}
\end{algorithm}
\subsection{Model Structure}
The schematic of the JSCC encoder and the U-Net structure in CDDM is illustrated in Figure \ref{CDESC}. In the JSCC encoder, the initial module is the patch embedding, responsible for partitioning the source image into non-overlapping patches. Subsequently, $M$ Basicblocks are employed to extract the semantic features from the source image. The $i$-th basicblock consists of a patch merging module and $N_i$ Swim Transformers, where $i=1,2...M$. After addressed by a basicblock, the height and width of the features are halved, while the channel dimensions are increased to $P_i$. Finally, a convolution head(Conv Head) layer is adopted to compute the features as transmitted signal $\mathbf{x}$.
The structure of the JSCC decoder is identical to that of the JSCC encoder, with the exception that the downsample modules in the JSCC encoder are replaced with upsample modules.

The model structure of CDDM is predominantly based on the convolutional improved U-Net architecture \cite{Ronneberger}. Initially, $\mathbf{y}_r$ undergoes a convolution layer and then serves as the input of the U-Net. Subsequently, the output of U-Net is further processed by another convolutional layer to generate the final output $\mathbf{y}$. The U-Net is comprised of various components, including convolutional residual (Conv-Res) blocks \cite{Zagoruyko}, convolutional attention (Conv-Attn) blocks, Down-Sampling blocks, and Up-Sampling blocks. A Down-Sampling block is a convolutional layer that performs down-sampling and maintains the same number of input and output channels.
The Up-Sampling block consists of an interpolation layer followed by a convolutional layer. A Conv-Attn is an attention block commonly adopted in classic transformer \cite{Vaswani}, but with the notable distinction of employing convolutional layers as a replacement for fully-connected (FC) layers. The structure of Conv-Res is depicted in Fig. \ref{CDESC}. In comparison to the classic residual block, the Conv-Res block substitutes FC layers with convolutional layers. Moreover, an additional convolutional layer is incorporated into the residual path to adjust the data dimension and enhance the model's capacity. The sampling step $t$ is addressed by a MLP and the embedded within the middle of the Conv-Res block. Multiple instances of these blocks are sequentially connected incorping two additional residual paths, ultimately forming the U-Net architecture.

\section{EXPERIMENTS RESULTS}
In this section, we provide a detailed description of the experimental setup and presented a considerable amount of experimental results, which comprehensively demonstrate the effectiveness of our proposed CDDM system. Additionally, we assess its natural robustness to channel estimation errors and its adaptability to different SNRs. 
\subsection{Experiment Setup}

\textbf{Datesets}: To obtain comprehensive and universally applicable results, we train and evaluate the proposed joint CDDM and JSCC system on two image datasets. 
CIFAR10 \cite{Krizhevsky} dataset is employed for low-resolution images with dimensions of $32\times 32$, comprising of 50000 color images for training and 10000 images for testing. The high-resolution images are obtained from DIV2K dataset \cite{DIV2K}, which includes 800 images for training and 100 images for testing. These images are collected from a wide range of real-world scenes and have a uniform resolution of 2K. 
During the training process, the images with high resolution are randomly cropped into patches with a size of $256\times256$.

\begin{figure}[t]
  \begin{center}
    \includegraphics*[width=8.5cm]{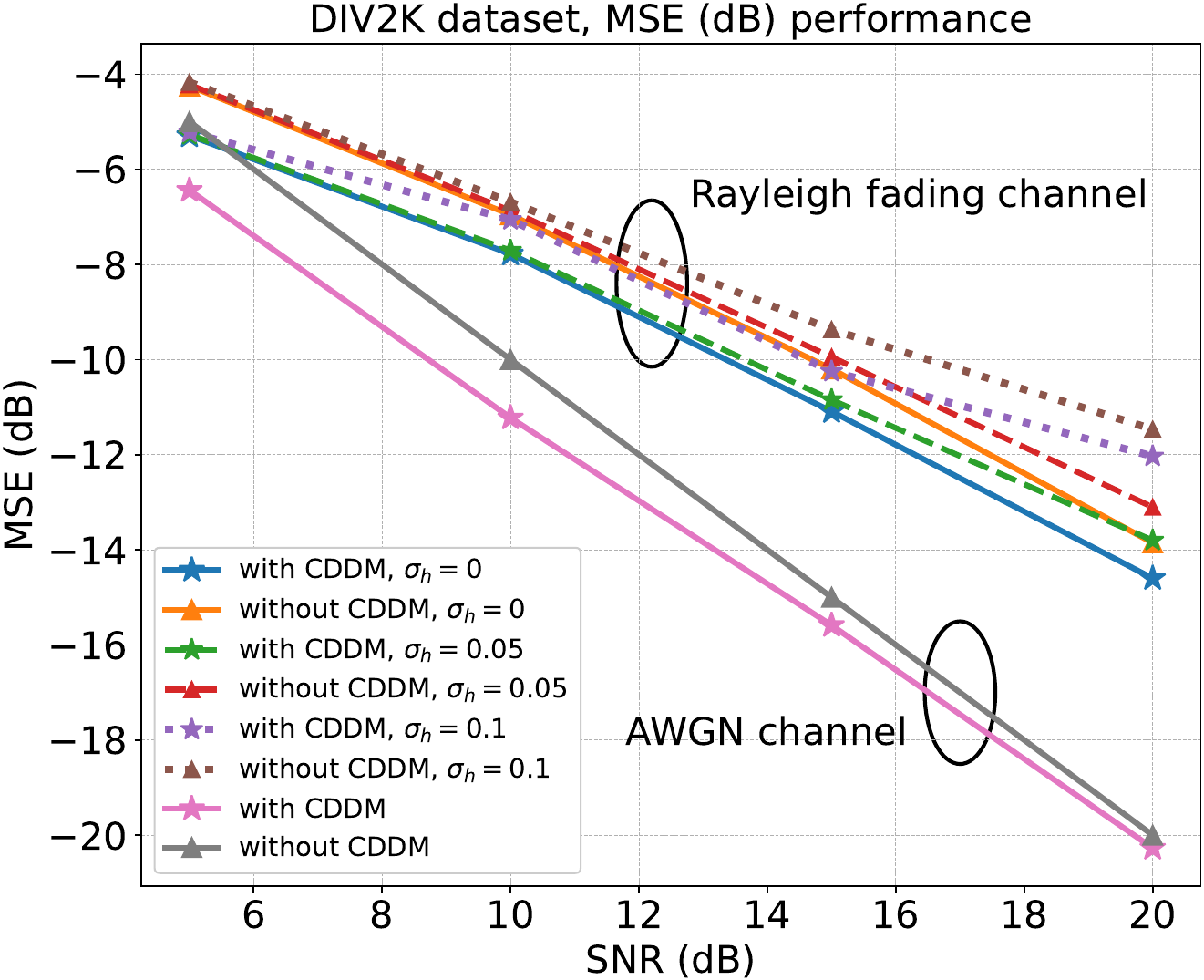}
  \end{center}
    \caption{{MSE performance of DIV2K versus SNRs under AWGN and Rayleigh fading channel with or without channel estimation errors. The CBR is $3/128$.}}
    \label{MSE}
\end{figure}

\textbf{Comparsion schemes}: We conduct a comparative analysis between the proposed joint CDDM and JSCC system and two other systems: the JSCC system without CDDM and the classical handcrafted separation-based source and channel coding system.
More specifically, the JSCC system shares an identical structure and training configuration within the joint CDDM and JSCC system. It is worth emphasizing that in the event of a change in channel SNR, both systems undergo retraining to optimize their performance under the specific SNR condition.
For the classical system, we employ the JPEG2000 codec for compression and LDPC\cite{DVB} codec for channel coding, marking as ``JPEG2000+LDPC". Here, we consider DVB-T2 LDPC codes with a block length of 64800 bits for different coding rates and quadrature amplitude modulations (QAM) adapted to the channel conditions.

\textbf{Evaluation Metrics}: We qualify the performance of all three schemes with both PSNR and MSSSIM. PSNR is a widely-used pixel-wise metric that measures the visiblity of errors between the reconstructed image and the reference image. A higher PSNR value indicates a smaller loss in the image quality. In this case, we adopt MSE to calculate $d(\cdot)$ during optimizating our networks. 
MSSSIM is a perceptual metric that specially concentrates on the structural similarity and content of images, which aligns more closely with the evaluation results of the human visual system (HVS). The multi-scale design allows it to demonstrate consistent performance across images with varying resolutions. 
The value of MSSSIM ranges from 0 to 1, where a higher value indicates a higher similarity to the reference image. Also in this case, we adopt 1-MSSSIM to calculate $d(\cdot)$ during optimizating our networks. When testing the performance, we convert MSSSIM into the form of dB for more intuitive observation and compaision. The formula is $MSSSIM\ (dB)=-10\ {\log}_{10}(1-MSSSIM)$. 

\begin{figure*}[t]
  \begin{center}
    \includegraphics*[width=18cm]{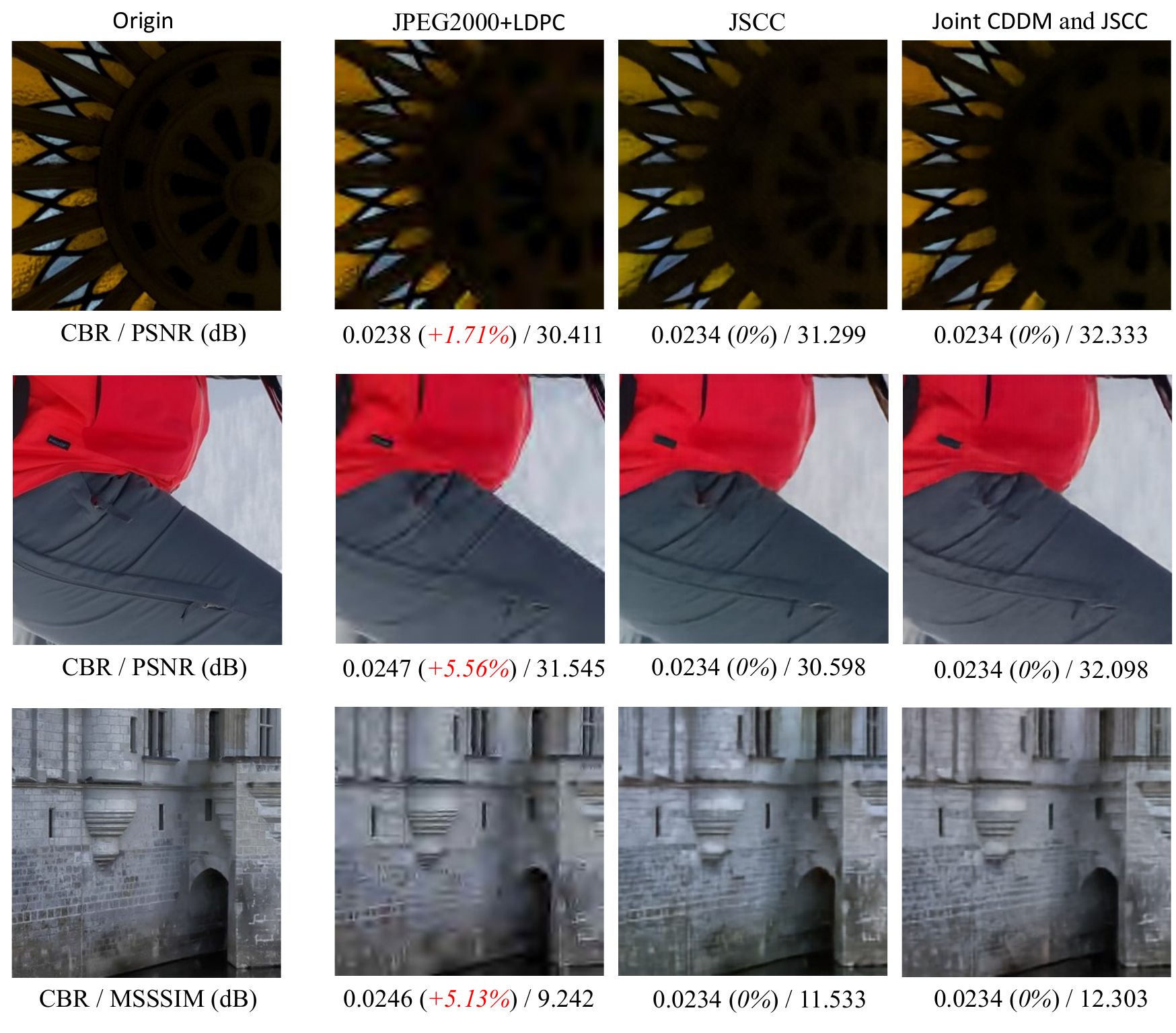}
  \end{center}
    \caption{{Examples of visualization results under Rayleigh fading channel at SNR=$10$ dB. The four columns display the original images and the reconstructed images obtained from their respective systems. The red number corresponds to the percentage of additional bandwidth cost in comparison to the joint CDDM and JSCC system.}}
    \label{visual}
\end{figure*}

\textbf{Training details}: For the CDDM training and sampling algorithms, we configure the parameter $T=1000$ and set $\alpha_t$ to constants decreasing linearly from an initial value of $\alpha_1=0.9999$ to a final value $\alpha_T=0.9800$. We set $t_{max}=93$ for CIFAR10 dataset while $t_{max}=52$ for DIV2K dataset. During optimizing the CDDM, we employ the Adma optimizer \cite{Diederik2} and implement a cosine warm-up learning rate schedule \cite{Loshchilov} with an initial learning rate of 0.0001.
In terms of the JSCC structure, the number of basic-blocks and patches varies depending on the dataset. For CIFAR10 dataset, the number of Basicblocks, denoted as $M$, is set to $2$, Swin Transformer numbers $[N_1,N_2]=[2,4]$ and channel dimensions $[P_1,P_2]=[128,256]$. On the other hand, for DIV2K dataset comprising high-resolution images, $M$ is set to $4$, Swin Transformer numbers $[N_1,N_2,N_3,N_4]=[2,2,6,2]$ and channel dimensions $[P_1,P_2,P_3,P_4]=[128,192,256,320]$. We employ Adam optimizer with a learning rate 0.0001 to optimize the JSCC \cite{KeYang}.

\begin{figure}[t]
  \begin{center}
    \includegraphics*[width=8.5cm]{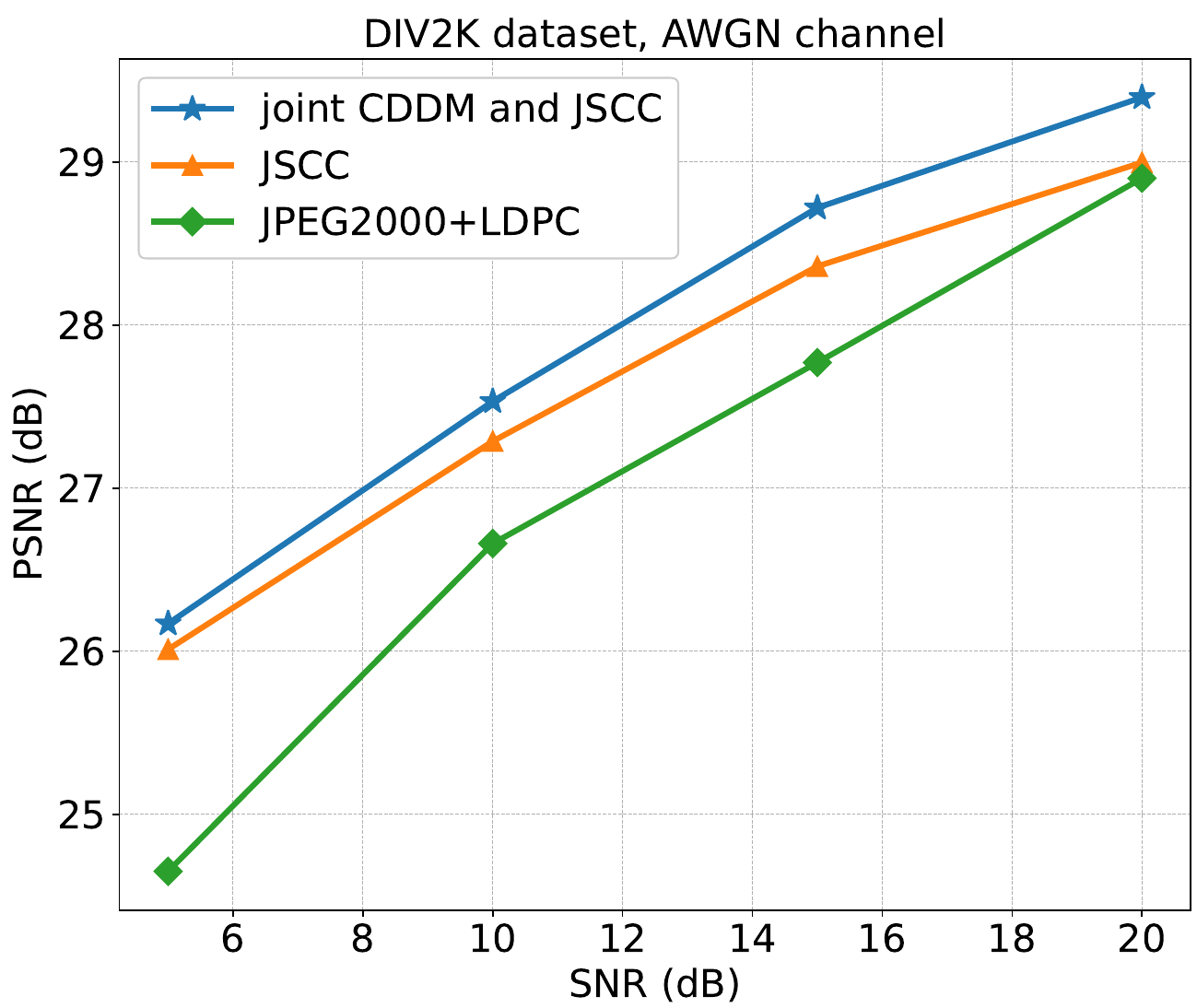}
  \end{center}
    \caption{{PSNR performance of DIV2K dataset versus SNR under AWGN channel. The CBR is set to $3/128$.}}
    \label{PSNR_AWGN_DIV2K_SNRs}
\end{figure}

\begin{figure}[t]
  \begin{center}
    \includegraphics*[width=8.5cm]{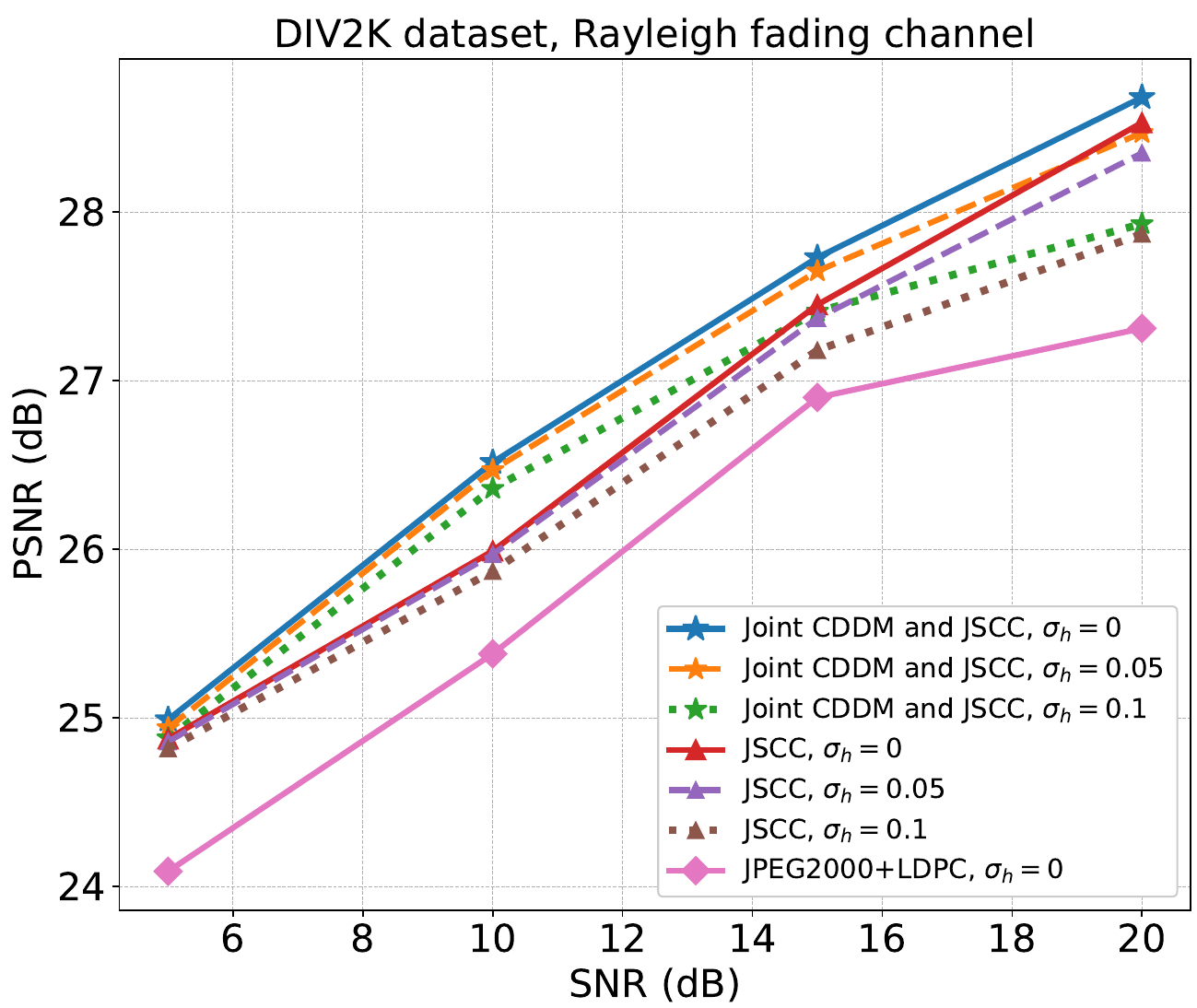}
  \end{center}
    \caption{{PSNR performance of DIV2K versus SNR under Rayleigh fading channel with or without channel estimation errors. The CBR is $3/128$.}}
    \label{PSNR_rayleigh_DIV2K_SNRs}
\end{figure}

\begin{figure}[t]
  \begin{center}
    \includegraphics*[width=8.5cm]{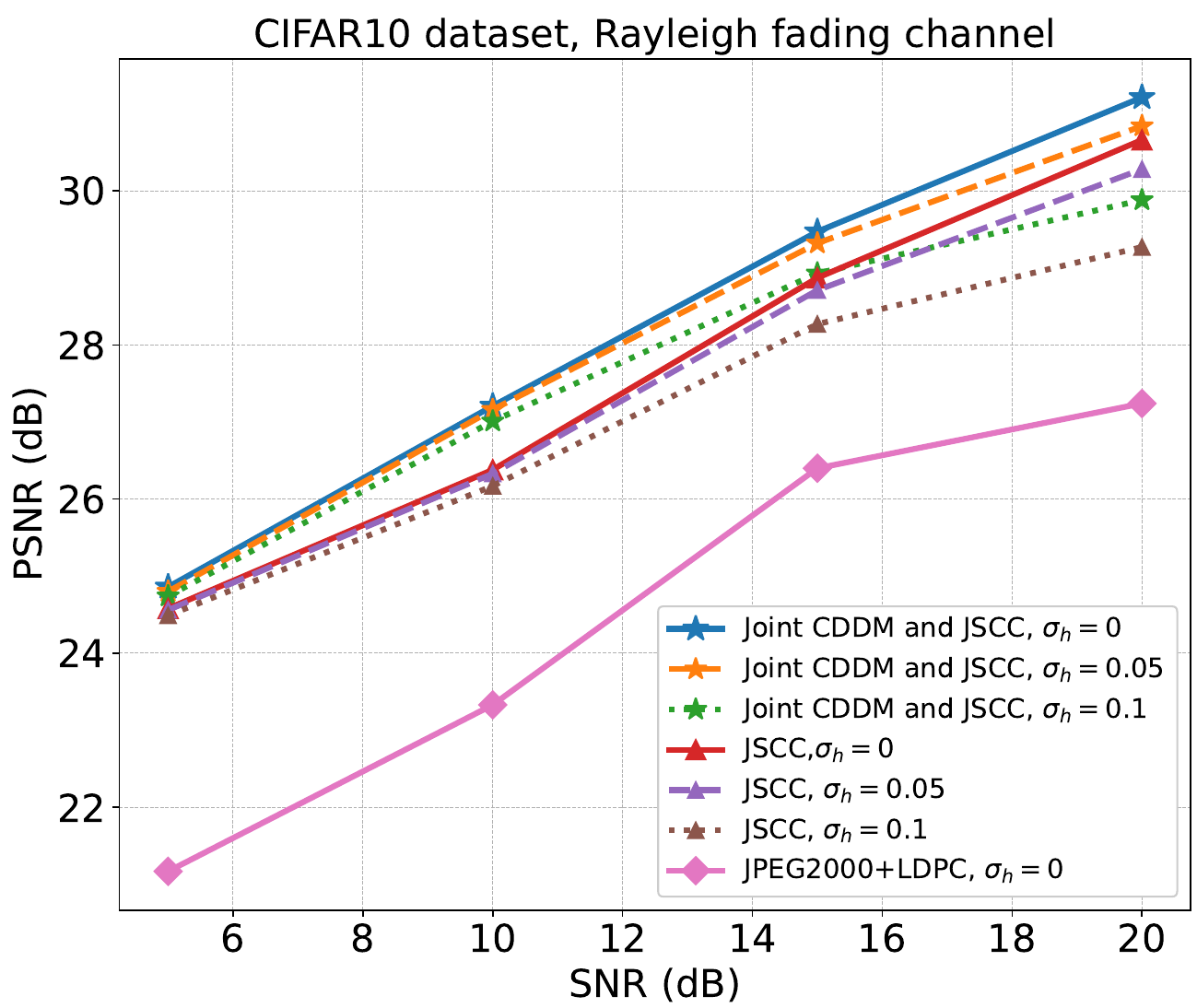}
  \end{center}
    \caption{{PSNR performance of CIFAR10 versus SNR under Rayleigh fading channel with or without channel estimation errors. The CBR is $1/8$.}}
    \label{PSNR_rayleigh_CIFAR10_SNRs}
\end{figure}

\begin{figure}[t]
  \begin{center}
    \includegraphics*[width=8.5cm]{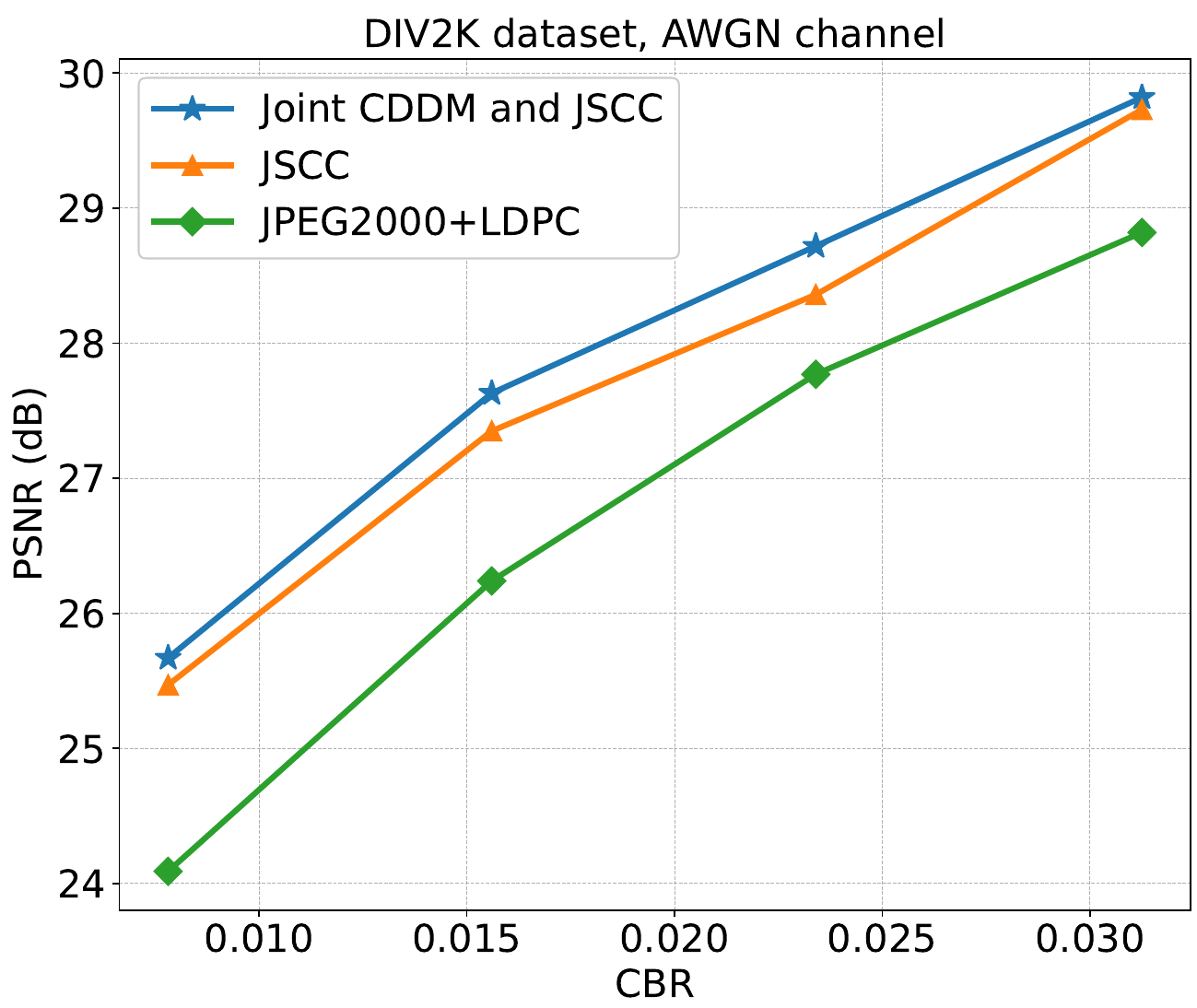}
  \end{center}
    \caption{{PSNR performance of DIV2K dataset versus CBR under AWGN channel. The SNR is $10$ dB.}}
    \label{PSNR_AWGN_DIV2K_CBRs}
\end{figure}

\subsection{MSE performance and visualization results}
Fig. \ref{MSE} illustrates the MSE performance of CDDM in different SNR regimes. The results are based on DIV2K dataset with JSCC trained for maximizing PSNR and channel bandwidth ratio (CBR) is set to $3/128$. In the case of using CDDM, we calculate the MSE between $\mathbf{x}$ and $\mathbf{y}$, while in the case of not using CDDM, we calculate the MSE between $\mathbf{x}$ and $\mathbf{y}_r$. As shown in Fig. \ref{CDESC}, $\mathbf{y}_r$ and $\mathbf{y}$ are the input and output of CDDM, respectively. The solid line in Fig. \ref{MSE} shows that the system with CDDM performs much better than the system without CDDM in all SNR regimes under both AWGN and Rayleigh fading channels. 
For example, for AWGN channel, the proposed CDDM reduces the MSE by $0.27$ dB at SNR=$20$ dB..
Meanwhile, it can be seen that as the SNR decreases, the gain of CDDM in MSE increases. This indicates that as the SNR decreases, i.e., the channel noise increases, the proposed CDDM is easier to remove more noise, e.g. $1.44$ dB gain at SNR=$5$ dB for AWGN channel. Moreover, it is important to note that under Rayleigh fading channel, MMSE has theoretically minimized the MSE, but CDDM can further reduce the MSE after MMSE. The reason for this fact is that CDDM can learn the distribution of $ \mathbf{x}_0=\mathbf{W_sx}$, and utilizes this learned knowledge to remove the noise thereby further reducing the MSE.

Additionally, to conduct a more comprehensively evaluation of our model, we assess the robustness of the proposed CDDM under Rayleigh fading channel with the presence of channel estimation errors. The receiver obtains a noisy estimation of $\mathbf{h}$, denoted as $\mathbf{\hat{h}}$ which is formulated as $\mathbf{\hat{h}}=\mathbf{h}+{\Delta}\mathbf{h}$, where ${\Delta}\mathbf{h}\sim \mathbb{CN} (0,\sigma_h^2\mathbf{I})$.
In Fig. \ref{MSE}, the dashed lines correspond to lower estimation errors with $\sigma_h=0.05$ and the dotted lines represent more estimation errors with $\sigma_h=0.1$.
It is observed that under $\sigma_h=0.05$, the joint CDDM and JSCC system maintains gains relative to perfect channel estimation across all SNR ranges. However, as $\sigma_h$ increases to $0.1$, the gains tend to decrease. This reduction is particularly notable at SNRs of $10$ and $20$ dB.

\begin{figure}[t]
  \begin{center}
    \includegraphics*[width=8.5cm]{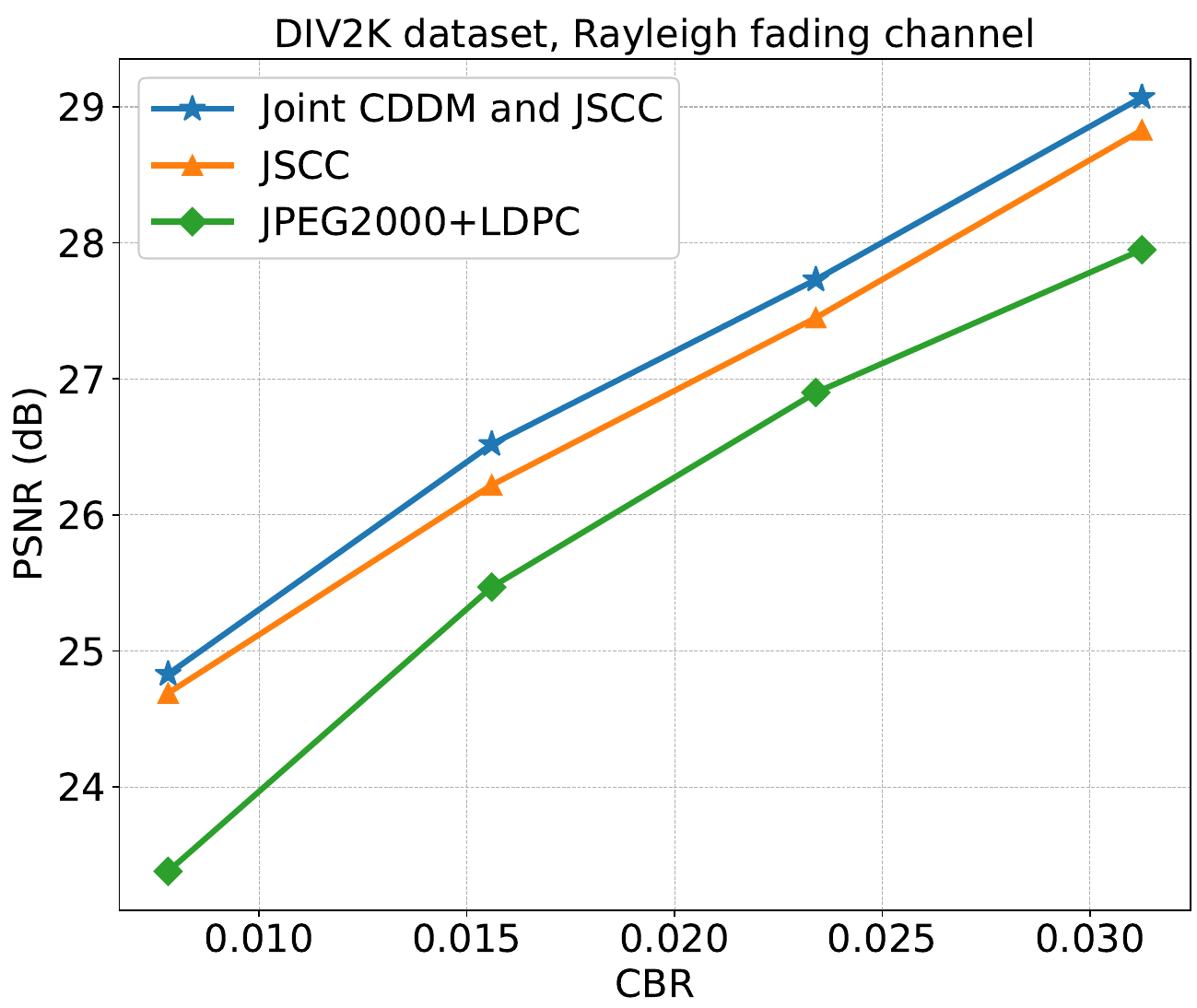}
  \end{center}
    \caption{{PSNR performance of DIV2K dataset versus CBR under Rayleigh fading channel. The SNR is $10$.}}
    \label{PSNR_rayleigh_DIV2K_CBRs}
\end{figure}

\begin{figure}[ht]
  \begin{center}
    \includegraphics*[width=8.5cm]{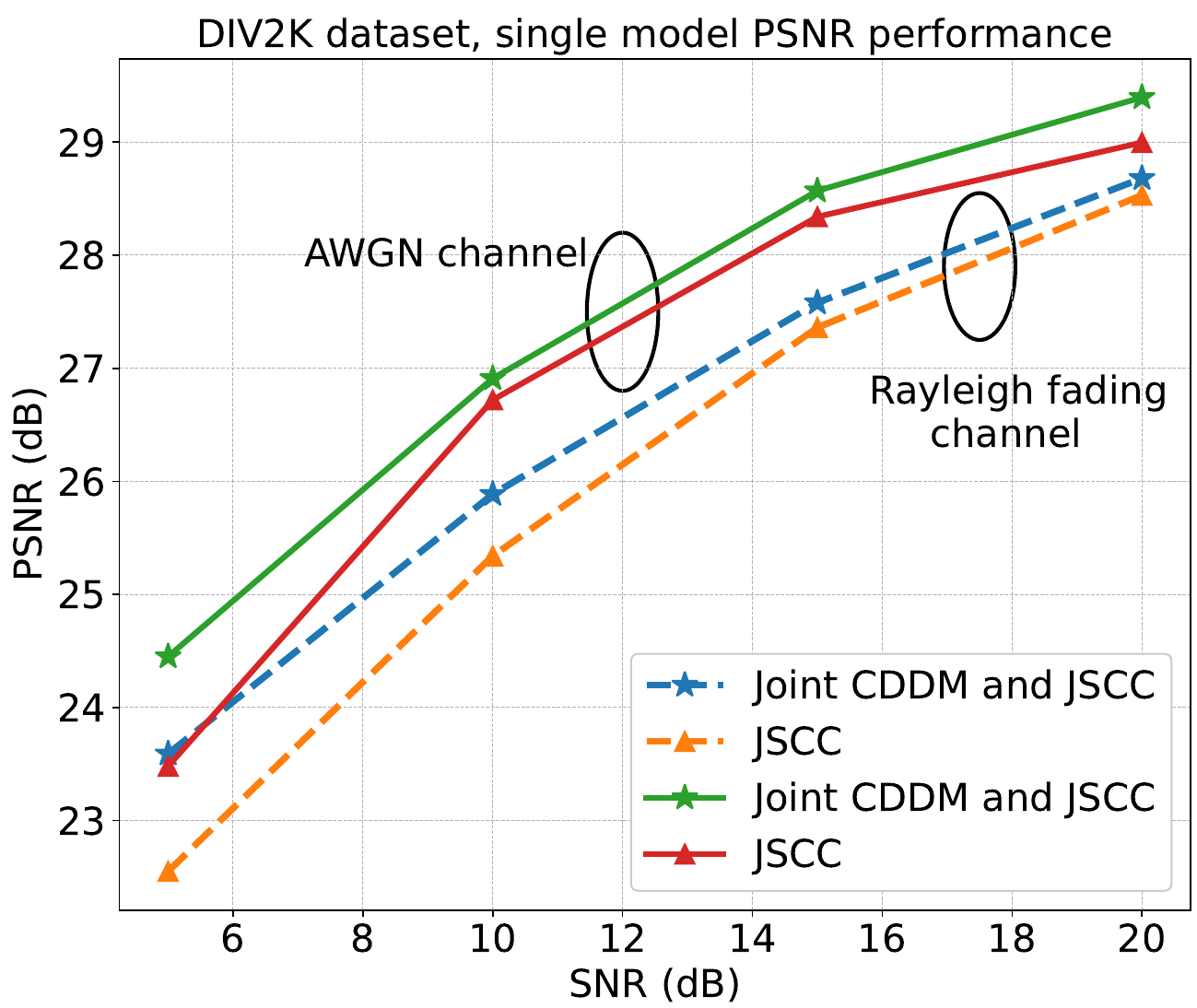}
  \end{center}
    \caption{{PSNR performance, trained at a SNR of $20$ dB, for DIV2K versus SNR under both AWGN and Rayleigh fading channels. The SNR is $10$ dB.}}
    \label{DIV2K_PSNR_single}
\end{figure}

\begin{figure}[ht]
  \begin{center}
    \includegraphics*[width=8.5cm]{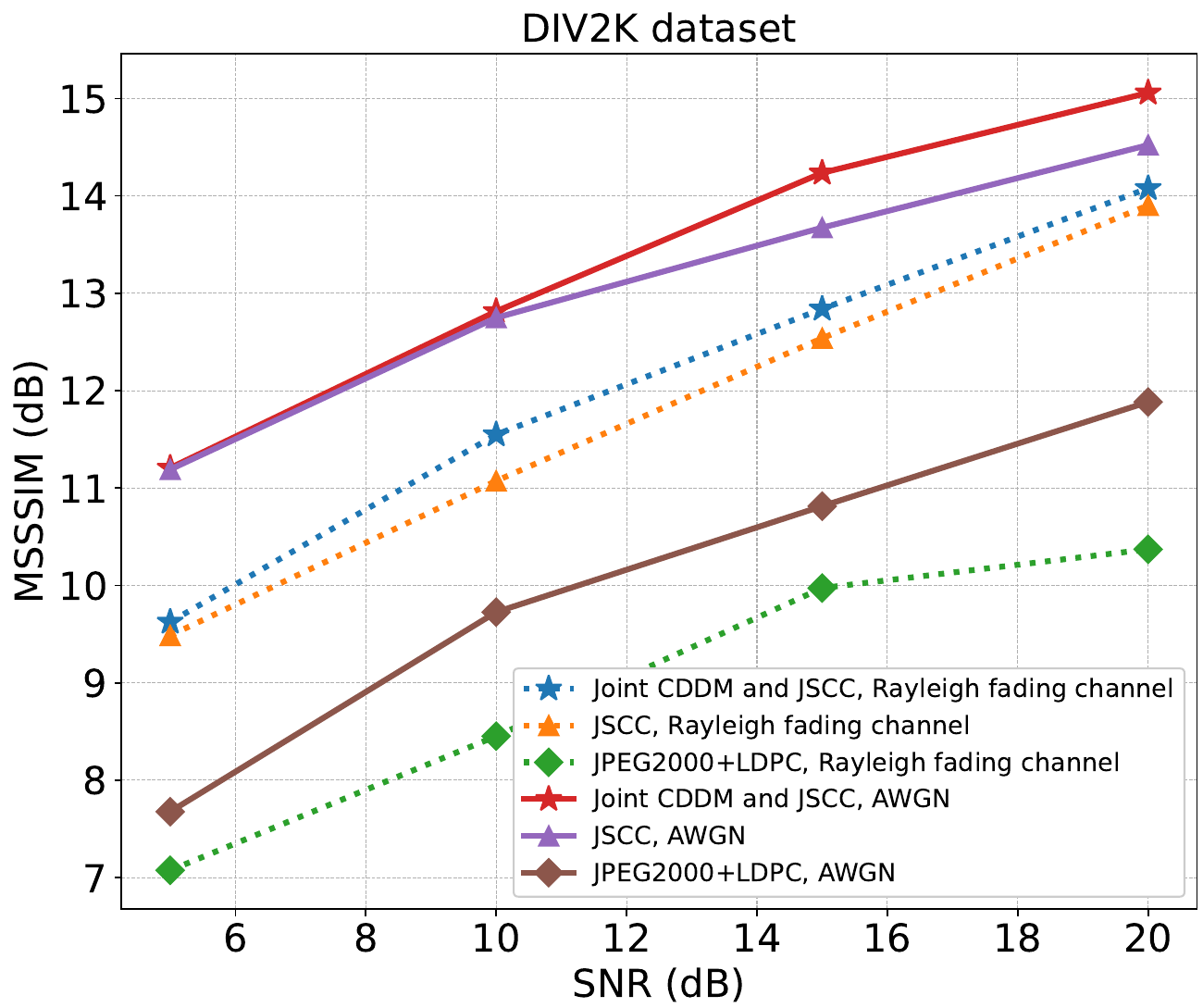}
  \end{center}
    \caption{{MSSSIM performance of DIV2K versus SNR under both AWGN and Rayleigh fading channels. The CBR is $3/128$.}}
    \label{DIV2K_MSSSIM_Rayleigh_SNRs}
\end{figure}

\begin{figure}[ht]
  \begin{center}
    \includegraphics*[width=8.5cm]{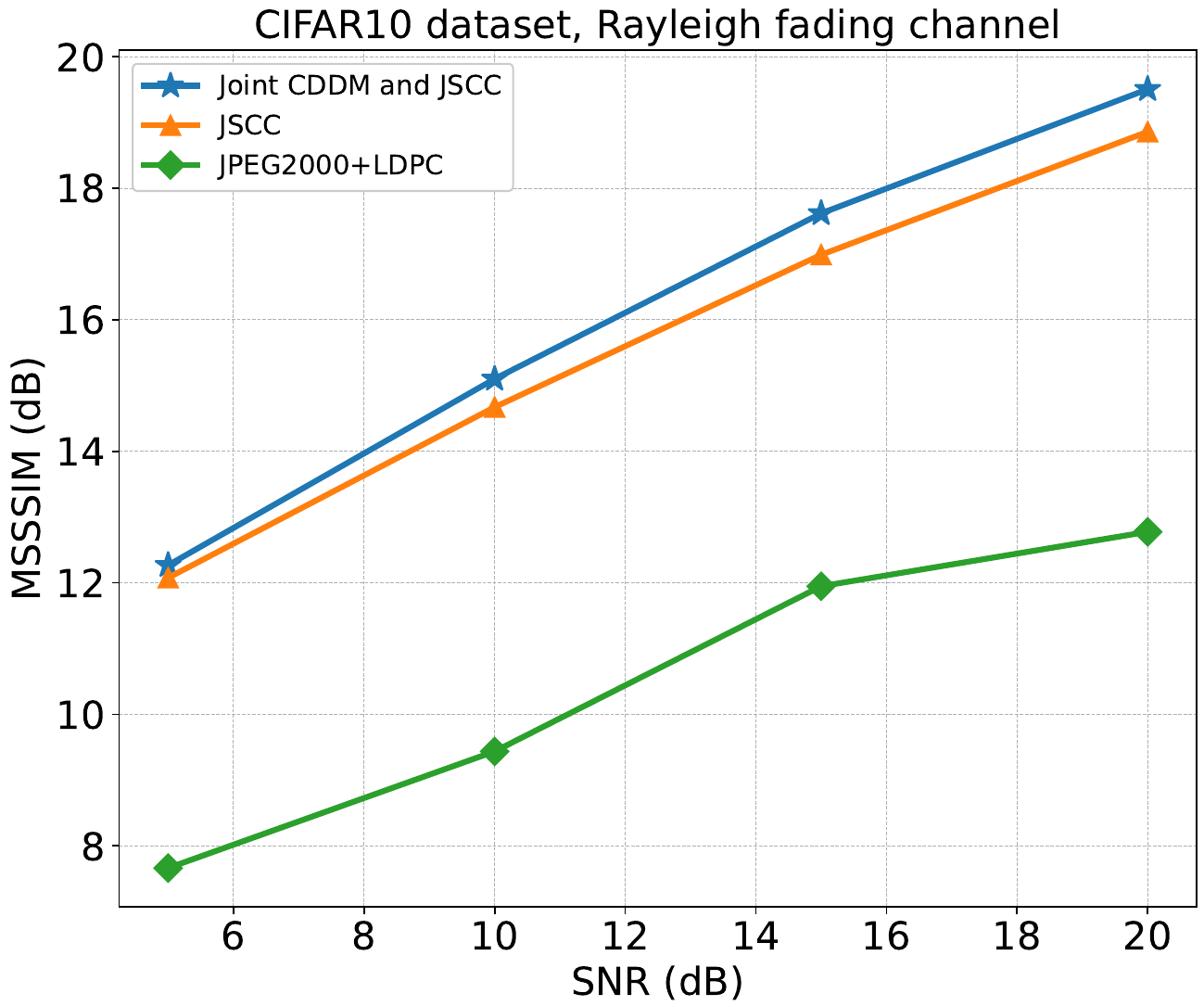}
  \end{center}
    \caption{{MSSSIM performance of CIFAR10 versus SNR under Rayleigh fading channel. The CBR is $1/8$.}}
    \label{CIFAR10_MSSSIM}
\end{figure}

\begin{figure}[ht]
  \begin{center}
    \includegraphics*[width=8.65cm]{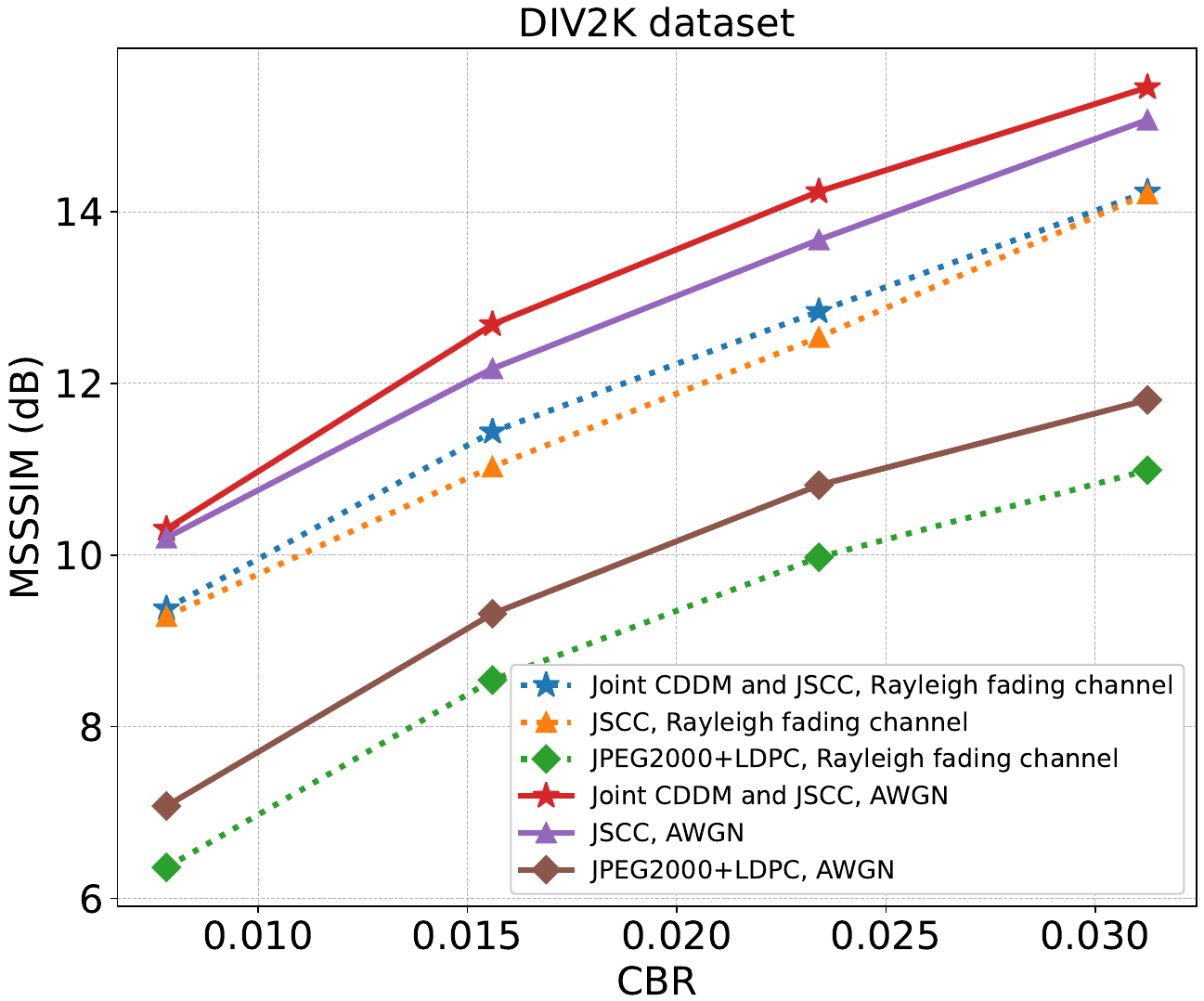}
  \end{center}
    \caption{{MSSSIM performance of DIV2K dataset versus CBR under both AWGN and Rayleigh fading channels. The SNR is $10$ dB.}}
    \label{DIV2K_rayleigh_MSSSIM_CBRs}
\end{figure}
Fig. \ref{visual} visualizes the reconstructions generated by the three systems. The results are obtained under Rayleigh fading channel with perfect channel estimation and an SNR of $10$ dB. It can be observed clearly that both JSCC-based systems outperform JPEG2000+LDPC in terms of visual quality, despite a slightly lower CBR. However, the reconstructed images obtained from the JSCC system demonstrate significant color aberration when compared to their corresponding original images. 
For example, the first image exhibits a lean towards a pale yellow hue, while the second and third images tend to lean towards a cyan color tone.
On the contrary, our joint CDDM and JSCC system simultaneously demonstrates superior color consistency and better visual quality.
\subsection{PSNR performance}
Fig. \ref{PSNR_AWGN_DIV2K_SNRs} illustrates the PSNR performance for DIV2K dataset versus SNR under AWGN channel. The CBR is configured to $3/128$. 
Our joint CDDM and JSCC system demonstrates superior performance compared to the JSCC system across a range of SNRs from $5$ to $20$ dB. Furthermore, the joint CDDM and JSCC system achieves significantly better performance when compared to the JPEG2000+LDPC system. Specifically, at an SNR of $20$ dB, the performance of the JPEG2000+LDPC system is comparable to that of the JSCC system, but still exhibits a $0.5$ dB inferiority compared to our joint CDDM and JSCC system.

Fig. \ref{PSNR_rayleigh_DIV2K_SNRs} and \ref{PSNR_rayleigh_CIFAR10_SNRs} illustrate the PSNR performance for both DIV2K and CIFAR10 datasets under Rayleigh fading channel. The CBR is $3/128$ for DIV2K and $1/8$ for CIFAR10. The solid line, dashed line and dotted line represent that $\sigma_h$ is $0$, $0.05$ and $0.1$, respectively. It can be observed that, the joint CDDM and JSCC system consistently outperforms the JSCC system across both datasets and all SNRs, i.e. $0.83$ dB for CIFAR10 dataset and $0.53$ dB for DIV2K dataset at SNR=$10$ dB with perfect channel estimation. Meanwhile, it is worth noting that the gain in PSNR performance for DIV2K dataset tends to decrease as the SNR increases when $\sigma_h=0.1$, which is aligns with the decrease in MSE performance gain.
The experimental results under both datasets, conducted at a channel estimation error level of $\sigma_h=0.1$, highlight the lack of natural robustness in our system when exposed to high channel estimation errors and high SNR conditions. This finding underscores the need to devise a specialized framework to mitigate the influence of channel estimation errors and enhance the system robustness in the future.

Fig. \ref{PSNR_AWGN_DIV2K_CBRs} and \ref{PSNR_rayleigh_DIV2K_CBRs} show the PSNR performance for DIV2K dataset in different CBRs under AWGN and Rayleigh fading channle, respectively. The SNR is set to $10$ dB. It is evident that our joint CDDM and JSCC system maintains effectiveness for complex high-resolution DIV2K dataset across various CBRs despite that the performance gain decreases as the CBR increases. This phenomenon can be attributed to the increase in the dimensionality of the transmitted signal $x$ when the CBR increases, thereby leading to a notable augmentation in the complexity of the learned distribution. However, to maintain experimental fairness, the structure and depth of the CDDM remain unchanged for different CBRs, consequently impeding the model's capacity to effectively learn the complex distribution and leading to a decline in performance gain.

Fig. \ref{DIV2K_PSNR_single} illustrates the PSNR performance versus SNR for DIV2K dataset over both AWGN and Rayleigh fading channel. In this experiment, the joint CDDM and JSCC system, as well as the JSCC system, are trained at a fixed SNR of $20$ dB and evaluated across various SNR values. It is evident that our joint CDDM and JSCC system consistently outperforms the JSCC system. More importantly, the performance gain becomes more pronounced as the SNR decreases in the Rayleigh fading channel. We attribute this phenomenon to the training of our CDDM utilizing Algorithm \ref{trainCDDM}, which encompasses a wide range of SNRs. Consequently, when the SNR varies, our CDDM still effectively reduces noise by adjusting the samlping step $m$, leading to enhanced performance. In contrast, the performance of the JSCC system deteriorates rapidly as the SNR decreases. This observation validates the adaptability of our joint CDDM and JSCC system to different SNRs.

\subsection{MSSSIM performance}


Fig. \ref{DIV2K_MSSSIM_Rayleigh_SNRs} shows the MSSSIM performance versus SNR for DIV2K dataset over both AWGN channel and Rayleigh fading channel. The solid lines represent performance under AWGN channel and the dotted lines represent performance under Rayleigh fading channel. The results demonstrate that under AWGN channel, our joint CDDM and JSCC system achieves a notable improvement in MSSSIM performance at SNRs of $15$ dB and $20$ dB particularly, i.e. $0.6$ dB at SNR=$15$ dB. At lower SNRs, we can still achieve an enhanced performance albeit with a quite small magnitude. Under Rayleigh fading channel, we achieve significant improvement across all SNRs.
Fig. \ref{CIFAR10_MSSSIM} demonstrates the MSSSIM performance for CIFAR10 dataset over Rayleigh fading channel. It can be observed that the joint CDDM and JSCC system outperforms both the JSCC system and the JPEG2000+LDPC system across all SNRs.


Fig. \ref{DIV2K_rayleigh_MSSSIM_CBRs} demonstrates the MSSSIM performance versus CBR for DIV2K under both AWGN channel and Rayleigh fading channel respectively. The results demonstrate that our joint CDDM and JSCC system outperforms the JSCC system under all examined conditions. Analogous to the PSNR performance, the magnitude of gain decrease when the CBR is large due to the same reason.
Moreover, all the experiment results conducted with MSSSIM performance show the consistent phenomenon that the MSSSIM performance of the JPEG2000+LDPC system is remarkably poor across all experimental configurations, showcasing a substantial disparity compared to both JSCC-based systems. These phenomenons prove that when considering the HVS, the JSCC system exhibits a dominant advantage over the JPEG2000+LDPC system. Furtherly, in this scenario, our joint CDDM and JSCC system can still enhance the performance.

The experiments conducted consistently demonstrate the efficacy of our joint CDDM and JSCC system, surpassing the performance of both the JSCC system and the JPEG2000+LDPC system across a wide range of conditions. These conditions encompass various SNRs, different CBRs, diverse evaluation metrics, distinct channel types and varying image resolutions.

\section{CONCLUSION}
In this paper, we have proposed the channel denoising diffusion models to eliminate the channel nosie under Rayleigh fading channel and AWGN channel. CDDM is trained utilizing a specialized noise schedule adapted to the wireless channel, which permits effective elimination of the channel noise via a suitable sampling algorithm in the reverse sampling process. Furtherly, we derived the sufficient condition under which our CDDM can reduce the conditional entropy of the received signal and demonstrate that the well-trained model satisfies this condition for smaller samlping steps through Monte Carlo experiments. CDDM is then applied into the semantic communications system based on JSCC. Extensive experimental results on CIFAR10 and DIV2K datasets show that under both AWGN and Rayleigh fading channels, the joint CDDM and JSCC system performs much better than the JSCC system without CDDM in terms of MSE, PSNR and MSSSIM.

\bibliography{ref}

\end{document}